\documentclass[11pt]{article}%

\usepackage{subcaption}%
\usepackage{amsmath}%
\usepackage{graphicx}%
\usepackage[cm]{fullpage}%
\usepackage{amssymb}%
\usepackage{xcolor}%
\usepackage{mleftright}%
\usepackage{scalerel}%
\usepackage{euscript}%
\usepackage{xspace}%
\usepackage{caption}%

\usepackage{titlesec}%
\titlelabel{\thetitle. }%

\usepackage{hyperref}%

\newcommand{\hrefb}[3][black]{\href{#2}{\color{#1}{#3}}}%

\providecommand{\IfPrinterVer}[2]{#2}%
\IfPrinterVer{%
   \usepackage{hyperref}%
}{%
   \usepackage{hyperref}%
   \hypersetup{%
      breaklinks,%
      colorlinks=true,%
      urlcolor=[rgb]{0.25,0.0,0.0},%
      linkcolor=[rgb]{0.5,0.0,0.0},%
      citecolor=[rgb]{0,0.2,0.445},%
      filecolor=[rgb]{0,0,0.4},
      anchorcolor=[rgb]={0.0,0.1,0.2}%
   }
}

\numberwithin{figure}{section}%
\numberwithin{table}{section}%
\numberwithin{equation}{section}%

\IfFileExists{sariel_computer.sty}{\def\sarielComp{1}}{}
\ifx\sarielComp\undefined%
\newcommand{\SarielComp}[1]{}
\newcommand{\NotSarielComp}[1]{#1}%
\else
\newcommand{\SarielComp}[1]{#1}%
\newcommand{\NotSarielComp}[1]{}%
\fi

\newcommand{\Of}{\Mh{\mathcal{O}}}%

\newcommand{\BSet}{\Mh{{B}}}%
\newcommand{\EBSet}{\Mh{{B^+}}}%

\usepackage[amsmath,thmmarks]{ntheorem}%
\theoremseparator{.}%

\theoremstyle{plain}%
\newtheorem{theorem}{Theorem}[section]

\newtheorem{lemma}[theorem]{Lemma}

\theoremstyle{plain}%
\theoremheaderfont{\sf} \theorembodyfont{\upshape}%
\newtheorem*{remark:unnumbered}[FakeCounter]{Remark}%
\newtheorem{definition}[theorem]{Definition}
\newtheorem*{defn:unnumbered}[FakeCounter]{Definition}

\newcommand{\myqedsymbol}{\rule{2mm}{2mm}}

\theoremheaderfont{\em}%
\theorembodyfont{\upshape}%
\theoremstyle{nonumberplain}%
\theoremseparator{}%
\theoremsymbol{\myqedsymbol}%
\newtheorem{proof}{Proof:}%

\newcommand{\eps}{\varepsilon}%
\newcommand{\epsA}{\Mh{\xi}}

\newcommand{\epsB}{\delta}

\newcommand{\epsR}{\Mh{\vartheta}}%

\newcommand{\ceil}[1]{\left\lceil {#1} \right\rceil}

\newcommand{\nextX}[1]{\mathrm{next}\pth{#1}}%
\newcommand{\prevX}[1]{\mathrm{prev}\pth{#1}}%

\newcommand{\GSet}{\mathcal{S}}

\newcommand{\HLinkShort}[2]{\hyperref[#2]{#1\ref*{#2}}}
\newcommand{\HLink}[2]{\hyperref[#2]{#1~\ref*{#2}}}
\newcommand{\HLinkPage}[2]{\hyperref[#2]{#1~\ref*{#2}%
      $_\text{p\pageref{#2}}$}}
\newcommand{\HLinkPageOnly}[1]{\hyperref[#1]{Page~\refpage*{#1}%
      $_\text{p\pageref{#1}}$}}

\newcommand{\HLinkSuffix}[3]{\hyperref[#2]{#1\ref*{#2}{#3}}}
\newcommand{\HLinkPageSuffix}[3]{\hyperref[#2]{#1\ref*{#2}%
      #3$_\text{p\pageref{#2}}$}}

\newcommand{\figlab}[1]{\label{fig:#1}}
\newcommand{\figref}[1]{\HLink{Figure}{fig:#1}}

\newcommand{\seclab}[1]{\label{sec:#1}}
\newcommand{\secref}[1]{\HLink{Section}{sec:#1}}

\newcommand{\remlab}[1]{\label{rem:#1}}
\newcommand{\defrefY}[2]{\hyperref[def:#2]{#1}}

\newcommand{\lemlab}[1]{\label{lemma:#1}}
\newcommand{\lemref}[1]{\HLink{Lemma}{lemma:#1}}%

\newcommand{\thmlab}[1]{{\label{theo:#1}}}
\newcommand{\thmref}[1]{\HLink{Theorem}{theo:#1}}

\providecommand{\eqlab}[1]{}%
\renewcommand{\eqlab}[1]{\label{equation:#1}}
\newcommand{\Eqref}[1]{\HLinkSuffix{Eq.~(}{equation:#1}{)}}

\providecommand{\Mh}[1]{{#1}}%

\IfFileExists{.latex_printer_friendly}{\def\GenPrinterVer{1}}{}%

\ifx\GenPrinterVer\undefined
   \IfFileExists{.latex_color}{\def\GenColorMath{1}}{}
\else
   \renewcommand{\IfPrinterVer}[2]{#1}%
\fi

\ifx\GenColorMath\undefined
\else
\renewcommand{\Mh}[1]{{\textcolor{red}{#1}}}%
\fi

\newcommand{\emphic}[2]{\textbf{\emph{#1}}}%
\newcommand{\emphi}[1]{\emphic{#1}{#1}}%

\SarielComp{%
\definecolor{blue25emph}{rgb}{0, 0, 11}
\renewcommand{\emphic}[2]{%
   \textcolor{blue25emph}{%
      \textbf{\emph{#1}}}%
   \index{#2}}

}

\usepackage{stmaryrd}%
\newcommand{\IntRange}[1]{\mleft[ #1 \mright]}
\newcommand{\IRX}[1]{\IntRange{#1}}%
\newcommand{\IRY}[2]{\left[ #1:#2 \right]}
\newcommand{\ILY}[2]{\left\llbracket #1, #2 \right\rrbracket}

\renewcommand{\Re}{\mathbb{R}}%

\newcommand{\PS}{\Mh{P}}%
\newcommand{\PB}{\Mh{B}}%
\newcommand{\PC}{\Mh{C}}%

\newcommand{\distSetY}[2]{\Mh{\mathsf{d}}\pth{#1,#2}}
\newcommand{\diamX}[1]{\mathrm{diam}\pth{#1}}%

\newcommand{\Graph}{\Mh{G}}%
\newcommand{\VV}{\Mh{V}}%
\newcommand{\Edges}{\Mh{E}}%
\newcommand{\EdgesX}[1]{\Mh{E}\pth{#1}}%

\newcommand{\Set}[2]{\left\{ #1 \;\middle\vert\; #2 \right\}}
\newcommand{\cardin}[1]{\left| {#1} \right|}%
\newcommand{\pth}[2][\!]{\mleft({#2}\mright)}%
\newcommand{\brc}[1]{\left\{ {#1} \right\}}

\newcommand{\dGY}[2]{\Mh{\mathsf{d}}\pth{#1,#2}}%
\newcommand{\dGZ}[3]{\Mh{\mathsf{d}_{#1}}\pth{#2,#3}}%
\newcommand{\dY}[2]{\left\| #1 - #2 \right\|}%

\newcommand{\pp}{\Mh{p}}%
\newcommand{\pq}{\Mh{q}}%
\newcommand{\pz}{\Mh{z}}%

\newcommand{\pout}{\hat{\pi}}%

\newcommand{\Gconst}{\Mh{\GraphA}}%
\newcommand{\Geps}{\Mh{\GraphA_{\epsR}}}%

\newcommand{\GraphA}{\Mh{H}}%

\newcommand{\etal}{\textit{et~al.}\xspace}

\newcommand{\ShadowC}{\Mh{\mathcal{S}}}%

\newcommand{\ShadowY}[2]{\Mh{\ShadowC}\pth{#1,#2}}%

\renewcommand{\th}{th\xspace}

\newcommand{\NbrX}[1]{\Mh{\Gamma}\pth{#1}}%
\newcommand{\IZ}[3]{\IA\pth{#1, #2, #3}}
\newcommand{\IA}{\Mh{\mathtt{I}}}%
\newcommand{\ExpZC}{\Mh{\mathrm{G_E}}}
\newcommand{\ExpZ}[3]{\Mh{\mathrm{G_E}}\pth{#1,#2,#3}}%
\newcommand{\ISet}{\Mh{\mathcal{I}}}%

\usepackage[inline]{enumitem}

\newlist{compactenumA}{enumerate}{5}%
\setlist[compactenumA]{topsep=0pt,itemsep=-1ex,partopsep=1ex,parsep=1ex,%
   label=(\Alph*)}%

\newlist{compactenuma}{enumerate}{5}%
\setlist[compactenuma]{topsep=0pt,itemsep=-1ex,partopsep=1ex,parsep=1ex,%
   label=(\alph*)}%

\newlist{compactenumI}{enumerate}{5}%
\setlist[compactenumI]{topsep=0pt,itemsep=-1ex,partopsep=1ex,parsep=1ex,%
   label=(\Roman*)}%

\newlist{compactenumi}{enumerate}{5}%
\setlist[compactenumi]{topsep=0pt,itemsep=-1ex,partopsep=1ex,parsep=1ex,%
   label=(\roman*)}%

\newlist{compactitem}{itemize}{5}%
\setlist[compactitem]{topsep=0pt,itemsep=-1ex,partopsep=1ex,parsep=1ex,\label=ensuremath{\bullet}}%

\newcommand{\DesSet}{\Mh{\EuScript{D}}}%
\newcommand{\AnsSet}{\Mh{\EuScript{A}}}%

\newcommand{\powTwoX}[1]{\mathrm{pow}_2\pth{#1}}%

\newcommand{\DesX}[1]{#1 \cap {\DesSet}}%

\newlength{\savedparindent}
\newcommand{\SaveIndent}{\setlength{\savedparindent}{\parindent}}

\newcommand{\nz}{\Mh{N}}%

\newcommand{\rightX}[1]{\Mh{\mathrm{right}}\pth{#1}}%
\newcommand{\leftX}[1]{\Mh{\mathrm{left}}\pth{#1}}%

\newcommand{\Term}[1]{\textsf{#1}}
\newcommand{\WSPD}{\Term{WSPD}\xspace}%

\newcommand{\SpreadC}{\Mh{\Phi}}%
\newcommand{\CPX}[1]{\Mh{\mathrm{c{}p}}\pth{#1}}%

\newcommand{\rchY}[2]{\Mh{R}_{#1}\pth{#2}}%
\newcommand{\orchX}[1]{\Mh{Q}_{#1}}%
\newcommand{\levelX}[1]{\Mh{\ell}\pth{#1}}%
\newcommand{\sizeX}[1]{\Mh{n}\pth{#1}}
\newcommand{\sizeBadX}[1]{\Mh{b}\pth{#1}}

\newcommand{\remove}[1]{}%

\newcommand{\GSB}{\Mh{\Graph_{\SpreadC}}}%
\newcommand{\Tree}{\Mh{T}}%

\newcommand{\cell}{\Box}%
\newcommand{\cellX}[1]{\cell_{#1}}%

\newcommand{\PcellX}[1]{\PS_{#1}}%
\newcommand{\WS}{\Mh{\mathcal{W}}}%

\newcommand{\Here}{\typeout{LOCATION: \currfilename\space L\the\inputlineno}\xspace}

\newcommand{\lossX}[1]{\Mh{\mathcal{L}}\pth{#1}}

\newcommand{\Daniel}{D\'aniel\xspace}%
\newcommand{\Olah}{Ol\'ah\xspace}%

\providecommand{\Mh}[1]{{#1}}%

\newcommand{\sLoc}{\Mh{s}}%
\newcommand{\tLoc}{\Mh{t}}%
\newcommand{\lShort}{\Mh{h}}%
\newcommand{\lLong}{\Mh{\ell}}%
\newcommand{\Shift}{\Mh{\CDelta}}%

\newcommand{\CDelta}{\Delta}

\newcommand{\ILeft}{\Mh{\mathsf{L}}}%
\newcommand{\IRight}{\Mh{\mathsf{R}}}%
\newcommand{\ballY}[2]{\mathrm{ball}\pth{#1, #2}}%

\newlength{\ppicX}
\newlength{\ppicY}

\newcommand{\order}{\sigma}
\newcommand{\orderset}{\Pi}
\newcommand{\ordAll}{\orderset^+}%

\newcommand{\constA}{\Mh{c}}%

\newcommand{\atgen}{\symbol{'100}}
\newcommand{\SarielThanks}[1]{\thanks{Department of Computer Science;
      University of Illinois; 201 N. Goodwin Avenue; Urbana, IL,
      61801, USA; {\tt sariel\atgen{}illinois.edu}; {\tt
         \url{http://sarielhp.org/}.} #1}}

\newcommand{\KevinThanks}[1]{%
   \thanks{%
      Department of Mathematics and Computing Science, TU Eindhoven,
      P.O. Box 513, 5600 MB Eindhoven, The Netherlands. %
      #1}%
}
\newcommand{\OlahThanks}[1]{%
   \thanks{%
      Department of Mathematics and Computing Science, TU Eindhoven,
      P.O. Box 513, 5600 MB Eindhoven, The Netherlands. %
      #1}%
}

\newcommand{\SaveContent}[2]{%
   \expandafter\newcommand{#1}{#2}%
}

\newcommand{\ts}{\hspace{0.6pt}}

\title{A Spanner for the Day After%
   \thanks{A preliminary version of the paper appeared in SoCG 2019
      \cite{bho-spda-19}.}%
}%

\author{%
   Kevin Buchin%
   \KevinThanks{}%
   \and%
   Sariel Har-Peled%
   \SarielThanks{Work on this paper was partially supported by NSF AF
      awards CCF-1421231 and CCF-1907400. %
   }%
   \and%
   \Daniel \Olah%
   \OlahThanks{Supported by the Netherlands Organisation for
      Scientific Research (N{W}O) through Gravitation-grant
      NETWORKS-024.002.003.}%
}

\date{\today}

\begin{document}
\maketitle

\begin{abstract}
    We show how to construct a $(1+\eps)$-spanner over a set $\PS$ of
    $n$ points in $\Re^d$ that is resilient to a catastrophic failure
    of nodes. Specifically, for prescribed parameters
    $\epsR,\eps \in (0,1)$, the computed spanner $\Graph$ has
    \begin{equation*}
        \Of\bigl(\eps^{-O(d)} \epsR^{-6} n \log n (\log\log n)^6
        \bigr)        
    \end{equation*}
    edges. Furthermore, for \emph{any} $k$, and \emph{any} deleted set
    $\BSet \subseteq \PS$ of $k$ points, the residual graph
    $\Graph \setminus \BSet$ is $(1+\eps)$-spanner for all the points
    of $\PS$ except for $(1+\epsR)k$ of them.  No previous
    constructions, beyond the trivial clique with $\Of(n^2)$ edges,
    were known with this resilience property (i.e., only a tiny
    additional fraction of vertices, $\epsR \cardin{\BSet}$, lose
    their distance preserving connectivity).

    Our construction works by first solving the exact problem in one
    dimension, and then showing a surprisingly simple and elegant
    construction in higher dimensions, that uses the one-dimensional
    construction in a black-box fashion.
\end{abstract}

\section{Introduction}

\paragraph{Spanners.} %
The vertices of a \emph{Euclidean graph} are points in $\Re^d$, and
the edges are weighted by the (Euclidean) distance between their
endpoints. Let $\Graph=(\PS,\Edges)$ be a Euclidean graph and
$\pp,\pq \in \PS$ be two vertices of $\Graph$. For a parameter
$t \geq 1$, a path between $\pp$ and $\pq$ in $\Graph$ is a
\emphi{$t$-path} if the length of the path is at most
$t \dY{\pp}{\pq}$, where $\dY{\pp}{\pq}$ is the Euclidean distance
between $\pp$ and $\pq$.  The graph $\Graph$ is a \emphi{$t$-spanner}
of $\PS$ if there is a $t$-path between any pair of points
$\pp,\pq\in \PS$.  Throughout the paper, $n$ denotes the cardinality
of the point set $\PS$, unless stated otherwise. We denote the length
of the shortest path between $\pp,\pq\in \PS$ in the graph $\Graph$ by
$\dGY{\pp}{\pq}$.

Spanners have been studied extensively. The main goal in spanner
constructions is to have small \emph{size}, that is, to use as few
edges as possible. Other desirable properties are low degrees
\cite{abcgh-sggsd-08,cc-srgs-10,s-gsfed-06}, low weight
\cite{bcfms-cgsnqt-10,gln-fgacsgs-02}, low diameter
\cite{ams-rdags-94,ams-dagss-99}, or to be resistant to failures. The
book by Narasimhan and Smid \cite{ns-gsn-07} gives a comprehensive
overview of geometric spanners.

\paragraph{Robustness.}
Here, our goal is to construct spanners that are robust according to
the notion introduced by Bose \etal \cite{bdms-rgs-13}.  Intuitively,
a spanner is robust if the deletion of $k$ vertices only harms a few
other vertices.  Formally, a graph $\Graph$ is an $f(k)$-robust
$t$-spanner, for some positive monotone function $f$, if for any set
$\BSet$ of $k$ vertices deleted in the graph, the remaining graph
$\Graph \setminus \BSet$ is still a $t$-spanner for at least $n-f(k)$
of the vertices. Note, that the graph $\Graph \setminus \BSet$ has
$n-k$ vertices -- namely, there are at most $\lossX{k} = f(k)-k$
additional vertices that no longer have good connectivity to the
remaining graph. The quantity $\lossX{k}$ is the \emphi{loss}. We are
interested in minimizing the loss.

The natural question is how many edges are needed to achieve a certain
robustness (since the clique has the desired property). That is, for a
given parameter $t$ and function $f$, what is the minimal size that is
needed to obtain an $f(k)$-robust $t$-spanner on any set of $n$
points.

A priori it is not clear that such a sparse graph should exist (for
$t$ a constant) for a point set in $\Re^d$, since the robustness
property looks quite strong.  Surprisingly,
Bose~\etal~\cite{bdms-rgs-13} showed that one can construct a
$\Of(k^2)$-robust $\Of(1)$-spanner with $\Of(n \log n)$ edges.  Bose
\etal \cite{bdms-rgs-13} proved various other bounds in the same vein
on the size of one-dimensional and higher-dimensional point sets.
Their most closely related result is that for the one-dimensional
point set, $\PS=\{1,2,\dots,n\}$ and for any $t \geq 1$, at least
$\Omega(n\log{n})$ edges are needed to construct an $\Of(k)$-robust
$t$-spanner.

An open problem left by Bose \etal~\cite{bdms-rgs-13} is the
construction of $\Of(k)$-robust spanners -- they only provide the easy
upper bound of $\Of(n^2)$ for this case.  In this paper, we present
several constructions for this case with optimal or near-optimal size.
These results even hold for a stronger requirement on the spanners,
defined next.

\paragraph{$\epsR$-reliable spanners.}
We are interested in building spanners where the loss is only
fractional. Specifically, given a parameter $\epsR$, we consider the
function $f(k) = (1+\epsR)k$. The loss in this case is
$\lossX{k} = f(k)-k = \epsR k$. A $(1+\epsR)k$-robust $t$-spanner is
\emphi{$\epsR$-reliable $t$-spanner}.

\paragraph{Exact reliable spanners.}
If the input point set is in one dimension, then one can easily
construct a $1$-spanner for the points, which means that the exact
distances between points on the line are preserved by the spanner --
indeed, simply connect the points from left to right. It becomes
significantly more challenging to construct such an exact spanner that
is reliable.

In case of robust spanners, for a function $f$ under some general conditions, Bose \etal \cite{bdms-rgs-13} gave a construction of exact $f(k)$-robust spanners. In particular, their result implies that one can construct $\Of(k \log k)$-robust $1$-spanners of size $\Of(n \log n)$ and, for any $\eps>0$, $\Of(k^{1+\eps})$-robust $1$-spanners of size $\Of(n \log\log n)$, for one-dimensional point sets.

\paragraph{Fault tolerant spanners.}
Robustness is not the only definition that captures the resistance of
a spanner network against vertex failures. A closely related notion is
fault tolerance \cite{lns-eacft-98,lns-iacft-02,l-nrftg-99}. A graph
$\Graph=(\PS,\Edges)$ is an \emph{$r$-fault tolerant $t$-spanner} if
for any set $\BSet$ of failed vertices with $\cardin{\BSet} \leq r$,
the graph $\Graph \setminus \BSet$ is still a $t$-spanner.  The
disadvantage of $r$-fault tolerance is that each vertex must have
degree at least $r+1$, otherwise the vertex can be isolated by
deleting its neighbors. Therefore, the graph has size at least
$\Omega(rn)$. There are constructions that show $\Of(rn)$ edges are
enough to build $r$-fault tolerant spanners. However, depending on the
chosen value $r$ the size can be too large.

In particular, fault tolerant spanners cannot have a near-linear
number of edges, and still withstand a widespread failure of
nodes. Specifically, if a fault tolerant spanner has $m$ edges, then
it can withstand a failure of at most $2m/n$ vertices. In sharp
contrast, $\epsR$-reliable spanners can withstand a widespread
failure. Indeed, a $\epsR$-reliable spanner can withstand a failure of
close to $n/(1+\epsR)$ of its vertices, and still have some vertices
that are connected by short paths in the remaining graph.

\paragraph{Region fault tolerant spanners.}
In a surprising result, Abam \etal \cite{adfg-rftgs-09} showed that
one can build a geometric spanner with near linear number of edges, so
that if the deleted set are all the points belonging to a convex
region (they also delete the edges intersecting this region), then the
residual graph is still a spanner for the remaining points.

\subsection{Our results} %
We investigate how to construct reliable spanners with very small loss
-- that is $\epsR$-reliable spanners. To the best of our knowledge
nothing was known on this case before this work.

\begin{compactenumA}

    \smallskip%
    \item \textbf{Exact $\Of(1)$-reliable spanner in one dimension.}
    Inspired by the reliability of constant degree expanders, we show
    how to construct an $\Of(1)$-reliable exact spanner on any
    one-dimensional set of $n$ points with at most $\Of(n \log{n})$
    edges.\footnote{This also improves an earlier preliminary
       construction by (some of) the authors \cite{bho-krsod-18}.} The
    idea of the construction is to build a binary tree over the
    points, and to build bipartite expanders between certain subsets
    of nodes in the same layer.  One can think of this construction as
    building different layers of expanders for different
    resolutions. The construction is described in
    \secref{1:spanner:1:d}. See \thmref{useless_constant} for the
    result.

    \smallskip%
    \item \textbf{Exact $\epsR$-reliable spanner in one dimension.}
    One can get added redundancy by systematically shifting the
    layers. Done carefully, this results in a $\epsR$-reliable exact
    spanner with $\Of(\epsR^{-6} n \log n)$ edges. The construction is
    described in \secref{eps:r:spanner:1:d}. See \thmref{useless_eps}
    for the result.

    \smallskip%
    \item \textbf{$\epsR$-reliable $(1+\eps)$-spanners in higher
       dimensions.} %
    We next show a simple construction of $\epsR$-reliable spanners in
    $\Re^d$, for $d$ being constant, using a recent result of Chan
    \etal \cite{chj-lsota-18}, which shows that one needs to maintain
    only a ``few'' linear orders. This immediately reduces the
    $d$-dimensional problem to maintaining a reliable spanner for each
    of these orderings, which is the problem we already solved. By
    applying a recursive scheme, using the same idea, we obtain the
    desired spanner of size
    $\Of\pth{\eps^{-O(d)} \epsR^{-6} n \log n (\log \log n)^{6}}$.
    See \secref{r:d:general} for details.

    \smallskip%
    \item \textbf{$\epsR$-reliable $(1+\eps$)-spanner with bounded
       spread.}  Since the general construction in $\Re^d$ has some
    additional factors that seem unnecessary, we present an optimal
    construction for the case when the point set has bounded spread,
    that is, the ratio of the largest and the smallest distance of
    point pairs is bounded by a polynomial of $n$.  Specifically, for
    points with spread $\SpreadC$ in $\Re^d$, for $d$ being constant,
    and for any $\eps>0$, we construct a $\epsR$-reliable
    $(1+\eps)$-spanner with
    $\Of\bigl( \eps^{-d} \epsR^{-2} n \log \SpreadC \bigr)$ edges. The
    basic idea is to construct a well-separated pair decomposition
    (\WSPD) directly on the quadtree of the point set, and convert
    every pair in the \WSPD into a reliable graph using a bipartite
    expander. The union of these graphs is the required reliable
    spanner.  See \secref{r:d:bounded:spread} and \lemref{G_2d-bdd}
    for details.
\end{compactenumA}

\paragraph{Shadow.} %
Underlying our construction is the notion of identifying the points
that loose connectivity when the failure set is removed. Intuitively,
a point is in the shadow if it is surrounded by failed points. We
believe that this concept is of independent interest -- see
\secref{shadow} for details and relevant results in one dimension.

\paragraph{The competition.} %
Independently of this work, Bose~\etal~\cite{bcdm-norgm-18} also
obtained an upper bound on the size of reliable spanners in
$\Re^d$. Their construction has $\Of(n \log^2n \log\log n)$ edges,
which is close, but worse, than our bound of
$\Of(n \log n (\log\log n)^6)$ edges. For both constructions, the
fundamental building blocks are expander graphs. However, the
construction in \cite{bcdm-norgm-18} does not use the one-dimensional
construction directly, but uses well-separated pair decomposition,
centroid decomposition and ideas for maintaining orders.

In particular, inspired by the preliminary work of Buchin \etal
\cite{bho-krsod-18}, in October 2018, Bose \etal %
\footnote{Available online here:
   \\%
\url{http://cglab.ca/~morin/publications/drafts/robust2/robust2-2018-10-29.pdf}.%
} achieved independently essentially the same one-dimensional results
(i.e., (A) and (B) above) by using expanders. Our own improved
one-dimensional results were announced earlier in a public talk on
June 12, 2018 (this was a Y{R}F talk in SoCG 2018).

To keep things in perspective, the results of
Bose~\etal~\cite{bcdm-norgm-18} are equivalent to our results (up to a
log factor), while being technically quite different (in two and
higher dimensions). In particular, their one-dimensional construction
has better dependency on the reliability parameter.  The main
advantage of our result, beyond the aforementioned (all important) log
factor, is that it is (conceptually) simpler -- the use of locality
sensitive orderings \cite{chj-lsota-18} makes all the difference.

\section{Preliminaries}
\seclab{prelims}

\subsection{Problem definition and notations}

Let $\IRX{n}$ denote the set $\{1,2,\dots,n\}$ and let
$\IRY{i}{j}=\{i,i+1,\dots,j\}$.

\begin{definition}[Robust spanner]
    Assume that we are given %
    \begin{compactenumi}
        \item a parameter $t \geq 1$,
        \item a graph $\Graph=(\PS,\Edges)$ that is a $t$-spanner,
        \item a function $f\colon \mathbb{N} \to \mathbb{R_+}$, and
        \item two point sets $\PS_1, \PS_2 \subseteq \PS$.
    \end{compactenumi}
    \smallskip%
    The graph $\Graph$ is an \emph{$f(k)$-robust $t$-spanner for
       $\PS_1 \times \PS_2$} if for any set of (failed) vertices
    $\BSet \subseteq \PS$ there exists a set $\EBSet\supseteq \BSet$
    with $\cardin{\EBSet} \leq f\pth{\cardin{\BSet}\bigr.}$ such that
    the subgraph
    \begin{equation*}
        \Graph \setminus \BSet%
        =%
        \Graph_{\PS \setminus \BSet}%
        =%
        \bigl( \PS
        \setminus \BSet, \Set{\bigl.uv \in \Edges(\Graph)}{u,v \in \PS
           \setminus \BSet}%
        \bigr)
    \end{equation*}
    induced by $\Graph$ on $\PS\setminus \BSet$ is a $t$-spanner for
    $(\PS_1 \setminus \EBSet) \times (\PS_2 \setminus \EBSet)$. That
    is, $\Graph \setminus \BSet$ has a $t$-path between all pairs of
    points $\pp \in \PS_1\setminus \EBSet$ and
    $\pq \in \PS_2 \setminus \EBSet$.  If $\PS_1 = \PS_2 = \PS$, then
    $\Graph$ is a \emphi{$f(k)$-robust $t$-spanner}.

    The vertices of $\EBSet \setminus \BSet$ are the vertices
    \emphi{harmed} by $\BSet$, and the quantity
    $\lossX{k} = f(k) -k \geq \cardin{\EBSet} - \cardin{\BSet}$ is the
    \emphi{loss}.
\end{definition}

\begin{definition}
    For a parameter $\epsR >0$, a graph $\Graph$ that is
    $(1+\epsR)k$-robust $t$-spanner is a \emphi{$\epsR$-reliable
       $t$-spanner}.
\end{definition}

\begin{definition}
    For a number $x > 0$, let $\powTwoX{x} = 2^{\ceil{\log x}}$ be the
    smallest number that is a power of $2$ and is at least as large as
    $x$.
\end{definition}

\subsection{Expanders}
\seclab{expander:c}

For a set $X$ of vertices in a graph $\Graph=(\VV,\Edges)$, let
\begin{equation*}
    \NbrX{X}%
    =%
    \Set{ v \in \VV}{ uv \in \Edges \text{ and } u \in X}    
\end{equation*}
be the \emphi{neighbors} of $X$ in $\Graph$. The following lemma,
which is a standard expander construction, provides the main building
block of our one-dimensional construction. We provide the proof only
for the sake of completeness, as this is well known. The survey by
Hoory \etal \cite{hlw-egta-06} gives a comprehensive study of
expanders.

\begin{lemma}
    \lemlab{expander}%
   Let $L,R$ be two disjoint sets, with a total of $n$ elements, and
   let $\epsA \in (0,1)$ be a parameter. One can build a bipartite
   graph $\Graph = (L \cup R, E)$ with $\Of(n/\epsA^2)$ edges, such
   that \smallskip%
   \begin{compactenumI}
       \item for any subset $X \subseteq L$, with
       $\cardin{X} \geq \epsA |L|$, we have that
       $\cardin{\NbrX{X}} > (1-\epsA)|R|$, and
       \smallskip%
       \item for any subset $Y \subseteq R$, with
       $\cardin{Y} \geq \epsA |R|$, we have that
       $\cardin{\NbrX{Y}} > (1-\epsA)|L|$.
   \end{compactenumI}
\end{lemma}

\begin{proof}
    This is a variant of an expander graph. See
    \cite[Section~5.3]{mr-ra-95} for a similar construction.

    Let $c=\ceil{3/\epsA^2}$. For every vertex in $L$, pick randomly
    and uniformly (with repetition) $\ell = c \ceil{n/|L|}$ neighbors
    in $R$. Do the same for every vertex in $R$, picking
    $c \ceil{n/|R|}$ neighbors at random from $L$. Let $\Graph$ be the
    resulting graph, after removing redundant parallel edges. Clearly,
    the number of edges is as required.

    As for the claimed properties,
    fix a subset $X \subseteq L$ of size at least $\epsA \cardin{L}$,
    and fix a subset on the right, $Z \subseteq R$ of size
    $\leq (1-\epsA)|R|$. Notice, that there are at most $2^n$ choices
    for both $X$ and $Z$. The probability that all the edges we picked
    for the vertices of $X$, stay inside $Z$, is at most
    \begin{equation*}
        (1-\epsA)^{\ell \cardin{X}}%
        \leq %
        (1-\epsA)^{\ell \epsA \cardin{L}}%
        \leq %
        (1-\epsA)^{c \epsA n}%
        \leq%
        \exp \pth{ - \epsA\cdot \frac{3}{\epsA^2} \cdot \epsA n}%
        =%
        \exp( -3  n)%
        \leq%
        1/8^n,
    \end{equation*}
    since $c \geq 3/\epsA^2$ and $1-\epsA \leq \exp(-\epsA)$. In
    particular, for a given $X$ the probability that this happens for
    any subset $Z$ is less than $2^n / 8^n = 1/4^n$. Thus, with
    probability less than $2^n / 4^n = 1/2^n$ there is an
    $X\subseteq L$ with $\NbrX{X} \leq (1-\epsA)n$. Using the same
    argument for $Y\subseteq R$ we get that the random graph does not
    have the desired properties with probability $2/2^n<1$ (for
    $n>1$). This implies that a graph with the desired properties
    exists.~
\end{proof}

\begin{remark*}
    (A) The above construction is randomized, but simple algebraic
    deterministic constructions are known
    \cite{ass-eccde-08,gg-eclss-81,rvw-ewzzg-02}.  One can improve
    their expansion ratio by repeatedly squaring them. These
    constructions are implicit, and are easily constructed quickly on
    the fly. As such, in the following, we assume that the expanders
    of \lemref{expander} are readily available for use in our
    constructions.

    (B) It is not hard to show that regular expanders (i.e., not
    bipartite as constructed above) are reliable graphs. That is,
    deleting a set of vertices leaves almost all the remaining
    vertices in a single connected component.
\end{remark*}

\section{Building reliable spanners in one dimension}
\seclab{one_dim}%

\subsection{Bounding the size of the shadow}
\seclab{shadow}

Our purpose is to build a reliable $1$-spanner in one
dimension. Intuitively, a point in $\IRX{n}$ is in trouble, if many of
its close by neighbors belong to the failure set $\BSet$. Such an
element is in the shadow of $\BSet$, defined formally next.

\begin{definition}
    Consider an arbitrary set $\BSet \subseteq \IRX{n}$ and a
    parameter $\alpha \in (0,1)$. A number $i$ is in the \emphi{left
       $\alpha$-shadow} of $\BSet$, if and only if there exists an
    integer $j \geq i$, such that
    \begin{math}
        \cardin{\IRY{i}{j} \cap \BSet \bigr. }%
        \geq%
        \alpha \cardin{\IRY{i}{j} \bigr.}.
    \end{math}
    Similarly, $i$ is in the \emphi{right $\alpha$-shadow} of $\BSet$,
    if and only if there exists an integer $h$, such that $h \leq i$
    and
    \begin{math}
        \cardin{\IRY{h}{i} \cap \BSet} \geq \alpha
        \cardin{\IRY{h}{i}}.
    \end{math}
    The \emphi{$\alpha$-shadow} of $\BSet$, denoted by
    $\ShadowY{\alpha}{\BSet}$, is the union of the left
    $\alpha$-shadow and the right $\alpha$-shadow.
    See \figref{shadow} for a visual interpretation of the shadow.
\begin{figure}
    \centerline{\includegraphics{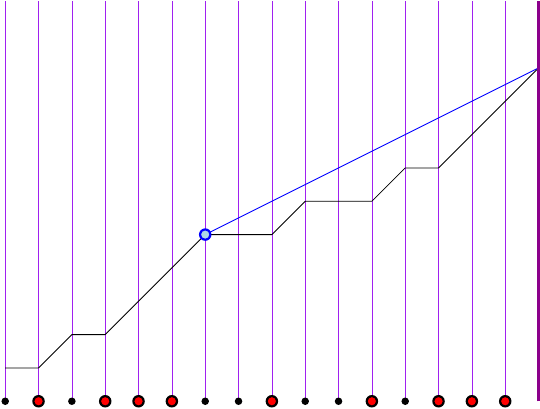}}
    \caption{Consider a set $\BSet \subseteq \IRX{n}$, the points
       $p_i=(i,1 + \cardin{ [1,i-1] \cap \BSet})$, for $i=1,\ldots,n$,
       and the associated polyline $\sigma$ resulting from connecting
       these points from left to right (the squares are the points
       that belong to $\BSet$).  A number
       $i \in \IRX{n} \setminus \BSet$ is in the left $\alpha$-shadow
       of $\BSet$, if the ray emanating from $p_i$ to the right, with
       slope $\alpha$, hits $\sigma$.}
       \figlab{shadow}
\end{figure}
\end{definition}

\begin{lemma}
    \lemlab{shadow}%
    For any set $\BSet \subseteq \IRX{n}$, and $\alpha \in (0,1)$, we
    have that
    $\cardin{\ShadowY{\alpha}{\BSet}} \leq \bigl(1 +
    2\ceil{1/\alpha}\bigr)\cardin{\BSet}$.
\end{lemma}

\begin{proof}
    Fix a set $\BSet = \brc{\pq_1 < \cdots < \pq_b}$. Let
    $\PS_0 = \brc{\pp_1< \cdots< \pp_{n-b}}$ be all the points of
    $\IRX{n}\setminus \BSet$, where $b= \cardin{\BSet}$.  Let $H$ be
    an initially empty set. Consider the process, that in the $i$\th
    iteration, for $i =1,\ldots, b$, set $L_i$ (resp. $R_i$) to be the
    set of $\ceil{1/\alpha}$ largest (resp. smallest) numbers of
    $\PS_{i-1}$ that are smaller (resp. larger) than $\pq_i$. We set
    $H = H \cup L_i \cup R_i \cup \brc{\pq_i} $, and
    $\PS_{i} = \PS_{i-1} \setminus ( L_i \cup R_i)$.

    The claim is that if $k \in \IRX{n}$ is in (say) the left
    $\alpha$-shadow of $\BSet$, then $k \in H$. If $k \in \BSet$ then
    the claim is immediate, so assume that $k \notin \BSet$. Next, fix
    a witness interval $J = \IRY{k}{k'}$, such that
    $\cardin{J\cap \BSet} \geq \alpha \cardin{J}$. Observe that the
    above process handles all the elements of $J^* = \BSet \cap J$ in
    turn. When handling an element $q \in J^*$, as long as
    $k \notin H$, the process either adds $k$ to $H$, or adds at least
    $1 + \ceil{1/\alpha}$ new points of $J$ (that were not previously
    in $H$) to $H$ -- indeed, the last $1 + \ceil{1/\alpha}$ ``free''
    points in the interval $\IRY{k}{q}$ are added to $H$.  Suppose,
    that $k \notin H$ at the end of the process. Then, we have
    \begin{equation*}
        \cardin{J} %
        \geq%
        \cardin{H \cap J}%
        \geq%
        \pth{1+ \ceil{1/\alpha}} \alpha \cardin{J}
        >%
        \cardin{J},
    \end{equation*}
     which is a contradiction.

    Applying the above argument symmetrically implies that, at the end
    of the process, $H$ contains all the points in the right and left
    shadows of $\BSet$. 
\end{proof}

One can get a better bound if $\alpha$ is close to one, as testified
by the following. This is crucial to achieve $\epsR$-reliability, since \lemref{shadow} is not sharp when $\alpha$ is close to one.

\begin{lemma}
    \lemlab{shadow_eps}%
    Fix a set $\BSet \subseteq \IRX{n}$, let $\alpha \in (2/3,1)$ be a
    parameter, and let $\ShadowY{\alpha}{\BSet}$ be the set of
    elements in the $\alpha$-shadow of $\BSet$.  We have that
    $\cardin{\ShadowY{\alpha}{\BSet}} \leq \cardin{\BSet} /
    (2\alpha-1)$.
\end{lemma}
\begin{proof}
    Let $c= 1-1/\alpha < 0$.  For $i=1,\ldots, n$, let $x_i = c$ if
    $i \in \BSet$, and $x_i=1$ otherwise.  For any interval $I$ of
    length $\CDelta$, with $\tau \CDelta$ elements in $\BSet$, such that
    $x(I) = \sum_{i \in I} x_i \leq 0$, we have that
    \begin{align*}
      x(I) \leq 0%
      &\iff%
        (1-\tau) \CDelta + c \tau \CDelta \leq 0%
        \iff%
        1-\tau \leq -\tau c%
        \iff%
        1/\tau \leq 1-c%
      \\&%
      \iff%
      1/\tau \leq 1-(1-1/\alpha)%
      \iff%
      {1}/{\tau} \leq {1}/{\alpha}%
      \iff%
      \tau \geq \alpha.
    \end{align*}
    An element $j \in \IRX{n} $ is in the left $\alpha$-shadow of
    $\BSet$ if and only if there exists an integer $j'$, such that
    $\cardin{\IRY{j}{j'} \cap \BSet} \geq \alpha \cardin{\IRY{j}{j'}}$
    and, by the above,
    \begin{math}
        x\bigl(\IRY{j}{j'}\bigr) \leq 0.
    \end{math}
    Namely, an integer $j$ in the left $\alpha$-shadow of $\BSet$
    corresponds to some prefix sum of the $x_i$s that starts at $j$
    and add up to some non-positive sum.  From this point on, we work
    with the sequence of numbers $x_1, \ldots, x_n$, using the above
    summation criterion to detect the elements in the left
    $\alpha$-shadow.

    For a location $j \in \IRX{n}$ that is in the left
    $\alpha$-shadow, let $W_{j} = \IRY{j}{j'}$ be the \emph{witness
       interval} for $j$ -- this is the shortest interval that has a
    non-positive sum that starts at $j$.  Let
    $I = W_{k} = \IRY{k}{k'}$ be the shortest witness interval, for
    any number in $\ShadowY{\alpha}{\BSet} \setminus \BSet$. For any
    $j \in \IRY{k+1}{k'}$, we have
    \begin{math}
        x\pth{\bigl. \IRY{k}{j-1}} + x\pth{\IRY{j}{k'}}%
        =%
        x(\IRY{k}{k'})%
        \leq%
        0.
    \end{math}
    Thus, if $x_j =1$, this implies that either $j$ or $k$ have
    shorter witness intervals than $I$, which is a contradiction to
    the choice of $k$.  We conclude that $x_j < 0$ for all
    $j \in \IRY{k+1}{k'}$, that is, $\IRY{k+1}{k'} \subseteq \BSet$.

    Letting $\ell = \cardin{I} = k'-k+1$, we have that
    \begin{equation*}
        \frac{\ell-1}{\ell} \geq \alpha%
        \iff%
        \ell-1 \geq \alpha\ell%
        \iff%
        \ell \geq \frac{1}{1-\alpha}%
        \iff%
        \ell \geq \ceil{\frac{1}{1-\alpha}} \geq 3,
    \end{equation*}
    as $\alpha \geq 2/3$. In particular, by the minimality of $I$, we
    have that $\ell = \ceil{1/(1-\alpha)}$.

    Let $J = \IRY{k}{k'-1} \subset I$. We have that $x(J) > 0$. For
    any $j \in \ShadowY{\alpha}{\BSet} \setminus \BSet$, such that
    $j \neq k$, consider the witness interval $W_{j}$. If $j >k$, then
    $j > k'$, as all the elements of $I$, except $k$, are in
    $\BSet$. If $j < k$ and $j' \in J$, then
    $\tau = x( \IRY{k}{j'}) > 0$, which implies that
    \begin{math}
        x\bigl(\IRY{j}{k-1}\bigr) = x\bigl(W_{j}\bigr) -\tau < 0,
    \end{math}
    but this is a contradiction to the definition of $W_{j}$. Namely,
    all the witness intervals either avoids $J$, or contain it in
    their interior. Given a witness interval $W_{j}$, such that
    $J \subset W_{j}$, we have
    \begin{math}
        x(W_{j} \setminus J) = x(W_{j}) - x(J) < x(W_{j}) \leq 0,
    \end{math}
    since $x(J) > 0$.

    So consider the new sequence of numbers
    \begin{math}
        x^{}_{\IRX{n} \setminus J} = x_1, \ldots, x_{k-1}, x_{k'},
        \ldots x_n
    \end{math}
    resulting from removing the elements that corresponds to $J$ from
    the sequence. Reclassify which elements are in the left shadow in
    the new sequence. By the above, any element that was in the shadow
    before, is going to be in the new shadow.  As such, one can charge
    the element $k$, that is in the left shadow (but not in $\BSet$),
    to all the other elements of $J$ (that are all in
    $\BSet$). Applying this charging scheme inductively, charges all
    the elements in the left shadow (that are not in $\BSet$) to
    elements in $\BSet$.  We conclude that the number of elements in
    the left shadow of $\BSet$, that are not in $\BSet$ is bounded by
    \begin{equation*}
        \frac{\cardin{\BSet}}{\cardin{J} - 1}%
        =%
        \frac{\cardin{\BSet}}{\ell - 2}%
        =%
        \frac{\cardin{\BSet}}{\ceil{1/(1-\alpha)} -2}%
        \leq%
        \frac{1-\alpha}{{1} -2(1-\alpha)}\cardin{\BSet}%
        =%
        \frac{1-\alpha}{2\alpha-1}\cardin{\BSet}.
    \end{equation*}
    The above argument can be applied symmetrically to the right
    shadow. We conclude that
    \begin{equation*}
        \cardin{\ShadowY{\alpha}{\BSet}}%
        \leq%
        \cardin{\BSet} + 2\frac{1-\alpha}{2\alpha-1}\cardin{\BSet}
        =%
        \frac{2 \alpha -1 +2 -2\alpha}{2\alpha-1}\cardin{\BSet}
        =%
        \frac{\cardin{\BSet}}{2\alpha-1}.%
    \end{equation*} 
\end{proof}

\subsection{Construction of %
   \texorpdfstring{$\Of(1)$}{O(1)}-reliable %
   exact spanners %
   in one dimension}
\seclab{1:spanner:1:d}

\subsubsection{Constructing the graph \texorpdfstring{$\Gconst$}{H}}

Assume $n$ is a power of two, and consider building the natural full
binary tree $T$ with the numbers of $\IRX{n}$ as the leaves. Every
node $v$ of $T$ corresponds to an interval of numbers of the form
$\IRY{i}{j}$, which we refer to as the \emphi{block} of $v$, see
\figref{blocks}. Let $\ISet$ be the resulting set of all blocks. In
each level one can sort the blocks of the tree from left to right. Two
adjacent blocks of the same level are \emph{neighbors}. For a block
$I \in \ISet$, let $\nextX{I}$ and $\prevX{I}$ be the blocks (in the
same level) directly to the right and left of $I$, respectively.

We build the graph of \lemref{expander} with $\epsA={1}/{16}$ for any
two neighboring blocks in $\ISet$. Let $\Gconst$ be the resulting
graph when taking the union over all the sets of edges generated by
the above.

\begin{figure}[t]
    \centerline{\includegraphics{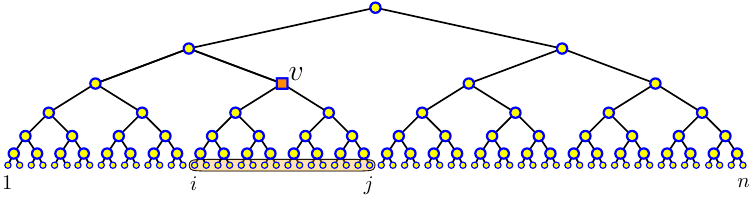}}%
    \caption{The binary tree built over $\IRX{n}$. The block of node
       $v$ is the interval $\IRY{i}{j}$.}
    \figlab{blocks}
\end{figure}

\subsubsection{Analysis}

Here, we show that the resulting graph $\Gconst$ is an $\Of(k)$-robust
$1$-spanner with $\Of(n\log n)$ edges.

\begin{lemma}
    \lemlab{G_1d-size}%
    The graph $\Gconst$ has $\Of(n\log{n})$ edges.
\end{lemma}
\begin{proof}
    Let $h=\log{n}$ be the depth of the tree $T$. For $i=1,2,\dots,h$,
    in the $i$\th level of $T$ there are $2^{h-i}$ nodes, and the
    blocks of these nodes have size $2^i$. The number of pairs of
    adjacent blocks in level $i$ is $2^{h-i}-1$ and each pair
    contributes $\Of(2^i)$ edges. Therefore, each level of $T$
    contributes $\Of(n)$ edges.  Summing this up for all levels
    implies the bound. 
\end{proof}

There is a natural path between two leaves in the tree $T$ going
through their lowest common ancestor. However, we need something
somewhat different here, as the path has to move forward (from left to
right) in the $1$-path.

Given two numbers $i$ and $j$, where $i<j$, consider the two blocks
$I, J \in \ISet$ that correspond to the two numbers at the bottom
level. Set $I_0 =I$, and $J_0 = J$.  We now describe a
\emphi{canonical} walk from $I$ to $J$, where initially
$\ell=0$. During the walk we have two active blocks $I_\ell$ and
$J_\ell$, that are both in the same level. For any block $I \in \ISet$
we denote its parent by $p(I)$. At every iteration we bring the two
active blocks closer to each other by moving up in the tree.

Specifically, repeatedly do the following:
\begin{compactenumA}
    \item If $I_\ell$ and $J_\ell$ are neighbors then the walk is
    done.

    \item If $I_\ell$ is the right child of $p(I_\ell)$, then set
    $I_{\ell+1} = \nextX{I_{\ell}}$ and $J_{\ell+1} = J_\ell$, and
    continue to the next iteration.

    \item If $J_\ell$ is the left child of $p(J_\ell)$, then set
    $I_{\ell+1} = {I_{\ell}}$ and $J_{\ell+1} = \prevX{J_\ell}$, and
    continue to the next iteration.

    \item Otherwise, algorithm ascends. It sets
    $I_{\ell+1} = p(I_\ell)$, and $I_{\ell+1} = p(J_\ell)$, and it
    continues to the next iteration.
\end{compactenumA}%
It is easy to verify that this walk is well defined, and let
\begin{equation*}
    \pi(i,j) \equiv
    \underbrace{I_0 \rightarrow I_1 \rightarrow \cdots \rightarrow
       I_\ell}_{\textsc{ascent}}
    \rightarrow%
    \underbrace{J_\ell \rightarrow \cdots \rightarrow J_0}_{\textsc{descent}}
\end{equation*}
be the resulting walk on the blocks where we removed repeated
blocks. \figref{canonical-path} illustrates the path of blocks between
two vertices $i$ and $j$.

\begin{figure}[t]
    \begin{subfigure}{\textwidth}
        \centerline{\includegraphics[page=2]%
           {figs/tree_blocks}}%
        \caption{The canonical path in the tree.}
    \end{subfigure}

    \begin{subfigure}{\textwidth}
        \centerline{\includegraphics[scale=0.7]{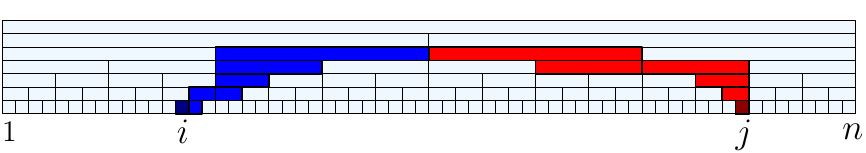}}
        \caption{The canonical path on the blocks.}
    \end{subfigure}
    \caption{The canonical path between the vertices $i$ and $j$ in
       two different representations. The blue nodes and blocks
       correspond to the ascent part and the red nodes and blocks
       correspond to the descent part of the walk.}
    \figlab{canonical-path}
\end{figure}

In the following, consider a fixed set $\BSet \subseteq \IRX{n}$ of
faulty nodes. A block $I \in \ISet$ is \emphi{$\alpha$-contaminated},
for some $\alpha\in (0,1)$, if
$\cardin{I \cap \BSet} \geq \alpha \cardin{I}$.

\begin{lemma}%
    \lemlab{contamination}%
    Consider two nodes $i,j \in \IRX{n}$, with $i<j$, and let
    $\pi(i,j)$ be the canonical path between $i$ and $j$. If any block
    of $\pi = \pi(i,j)$ is $\alpha$-contaminated, then $i$ or $j$ are
    in the $\alpha/3$-shadow of $\BSet$.
\end{lemma}

\begin{proof}
    Assume the contamination happens in the left half of the path,
    i.e., at some block $I_t$, during the ascent from $i$ to the
    connecting block to the descent path into $j$, see
    \figref{contamination}.
    \begin{figure}[ht]
        \centerline{%
           \includegraphics[width=0.95\linewidth]%
           {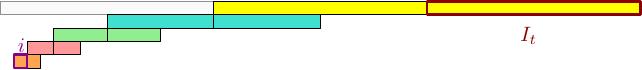}%
        }%
        \vspace{-0.08cm}%
        \caption{Contamination.}
        \figlab{contamination}
    \end{figure}
    By construction, there could be only one block before $I_t$ on the
    path of the same level, and all previous blocks are smaller, and
    there are at most two blocks at each level. Furthermore, for two
    consecutive $I_j, I_{j+1}$ that are blocks of different levels,
    $I_j \subseteq I_{j+1}$. Thus we have that either $i \in I_t$, or
    $i \in \prevX{I_t}$, or $i \in \prevX{\prevX{I_t}}$, since there
    are at most
    $\cardin{I_t}+ \cardin{I_t}/2+ \dots+2+1 = 2\cardin{I_t}-1$
    vertices that are contained in the path before the block
    $I_t$. Notice that if $i \in I_t$, then it is the leftmost point
    of $I_t$.

    So, let $r$ be the maximum number in $I_t$, and observe that
    $\cardin{\IRY{i}{r}} \leq 3\cardin{I_t}$. Furthermore, since $I_t$
    is $\alpha$-contaminated, we have
    \begin{equation*}
        \cardin{\IRY{i}{r} \cap \BSet}%
        \geq%
        \cardin{I_t \cap \BSet}%
        \geq%
        \alpha \cardin{I_t}%
        \geq%
        (\alpha/3)\cardin{\IRY{i}{r}}.        
    \end{equation*}
    Thus, the number $i$ is the $\alpha/3$-shadow, as claimed.

    The other case, when the contamination happens in the right part
    during the descent, is handled symmetrically. 
\end{proof}

\begin{theorem}
    \thmlab{useless_constant}%
    The graph $\Gconst$, constructed above, on the set $\IRX{n}$ is an
    $\Of(1)$-reliable exact spanner and has $\Of(n\log{n})$ edges.
\end{theorem}

\begin{proof}
    The size is proved in \lemref{G_1d-size}.  Let $\alpha =
    1/32$. Let $\EBSet$ be the set of vertices that are in the
    $\alpha/3$-shadow of $\BSet$, that is,
    $\EBSet = \ShadowY{\alpha/3}{\BSet}$. By \lemref{shadow} we
    have that
    \begin{math}
        \cardin{\EBSet} \leq (1+2\ceil{3/\alpha}) \cardin{\BSet} \leq
        200 \cardin{\BSet}.
    \end{math}

    Consider any two vertices $i,j \in \IRX{n} \setminus \EBSet$. Let
    $\pi(i,j)$ be the canonical path between $i$ and $j$. None of the
    blocks in this path are $\alpha$-contaminated, by
    \lemref{contamination}.

    Let $\GSet$ be the set of all vertices that have a $1$-path from
    $i$ to them (after removing the vertices of $\BSet$). Consider the
    ascent part of the path
    $\pi(i,j)\colon I_0 \rightarrow I_1 \rightarrow \cdots \rightarrow
    I_\ell$. The claim is that for every block $I_t$ in this path, we
    have that at least ${3}/{4}$ of the vertices have $1$-paths from
    $i$ (i.e.,
    $\cardin{I_t \cap \GSet} \geq \frac{3}{4} \cardin{I_t}$).

    This claim is proven by induction. The claim trivially holds for
    $I_0$. Now, consider two consecutive blocks
    $I_t \rightarrow I_{t+1}$.  There are two cases:
    \begin{compactenumi}
        \smallskip%
        \item $I_{t+1} = \nextX{I_t}$. Then, the graph $\Gconst$
        includes the expander graph on $I_t, I_{t+1}$ described in
        \lemref{expander}. At least $({3}/{4})\cardin{I_t}$ vertices
        of $I_t$ are in $\GSet$. As such, at least
        $\frac{15}{16}\cardin{I_{t+1}}$ vertices of $I_{t+1}$ are
        reachable from the vertices of $I_t$.  Since $I_{t+1}$ is not
        $\alpha$-contaminated, at most an $\alpha$-fraction of
        vertices of $I_{t+1}$ are in $\BSet$, and it follows that
        \begin{math}
            \cardin{I_{t+1} \cap \GSet}%
            \geq%
            (\frac{15}{16} - \alpha) \cardin{I_{t+1}}%
            \geq%
            \frac{3}{4} \cardin{I_{t+1}},
        \end{math}
        as claimed.

        \medskip%
        \item $I_{t+1}$ is the parent of $I_t$. In this case, $I_t$ is
        the left child of $I_{t+1}$. Let $I_t'$ be the right child of
        $I_{t+1}$. Since $I_{t+1}$ is not $\alpha$-contaminated, we
        have that
        $\cardin{I_{t+1} \cap \BSet} \leq \alpha\cardin{I_{t+1}}$. As
        such,
        \begin{equation*}
            \cardin{I_{t}' \cap \BSet}%
            \leq%
            \cardin{I_{t+1} \cap \BSet}%
            \leq%
            2\alpha\cardin{I_{t}'}
        \end{equation*}
        Now, by the expander construction on $(I_t,I_t')$, and arguing
        as above, we have
        \begin{equation*}
            \cardin{I_{t}' \cap \GSet} %
            \geq%
            \Bigl( \frac{15}{16} -2\alpha \Bigr) \cardin{I_t'}%
            \geq %
            \frac{3}{4} \cardin{I_{t}'},
        \end{equation*}
        which implies that
        $\cardin{I_{t+1} \cap \GSet} \geq \frac{3}{4}
        \cardin{I_{t+1}}$.
    \end{compactenumi}

    \medskip%
    The symmetric claim for the descent part of the path is handled in
    a similar fashion. Therefore, at least ${3}/{4}$ of the points in
    $J_\ell$ can reach $j$ with a $1$-path. Using these and the
    expander construction between $I_\ell$ and $J_\ell$, we conclude
    that there is a $1$-path from $i$ to $j$ in
    $\Gconst \setminus \BSet$, as claimed. 
\end{proof}

Note that it is easy to generalize the construction for arbitrary
$n$. Let $h$ be the integer such that $2^{h-1}< n < 2^h$ and build the
graph $\Gconst$ on $\{1,2,3,\dots,2^h\}$. Since $\Gconst$ is a
$1$-spanner, the $1$-paths between any pair of vertices of $\IRX{n}$
do not use any vertices from $\{n+1,\dots,2^h\}$. Therefore, we can
simply delete the part of $\Gconst$ that is beyond $n$ to obtain an
$\Of(1)$-reliable $1$-spanner on $\IRX{n}$.  Since we defined $\EBSet$
to be the shadow of $\BSet$, the $\Of(1)$-reliability is inherited
automatically.

We also note that no effort was made to optimize the constants in the
above construction.

\subsection{Construction of \texorpdfstring{$\epsR$}{theta}-reliable %
   exact spanners in one %
   dimension}
\seclab{eps:r:spanner:1:d}

Here, we show how to extend \thmref{useless_constant}, to build a
$1$-spanner on $\IRX{n}$, such that for any fixed $\epsR \in (0,1)$
and any set $\BSet$ of $k$ deleted vertices, at most $(1+\epsR)k$
vertices are no longer connected (by a $1$-path) after the removal of
$\BSet$.  The basic idea is to retrace the construction of
\thmref{useless_constant}, and extend it to this more challenging
case. There are two main new ingredients: (i) a shifting scheme, and (ii)
using much larger intervals when ascending from a level upward.
Unfortunately, the details are somewhat tedious.

\subsubsection{The construction}

Let $\IRX{n}$ be the ground set, and assume that $n$ is a power of
two, and let $h = \log n$. Let
\begin{equation}
    \nz = \powTwoX{c/\epsR^2}%
    \qquad\text{ and } \qquad%
    \epsA = \frac{1}{32N},
    \eqlab{n-z-val}%
\end{equation}
where $c$ is a sufficiently large constant (e.g., $c\geq512$).  We
first connect any $i \in \IRX{n}$, to all the vertices that are within
distance at most $3N$ from it, by adding an edge between the two
vertices. Let $\Graph_0$ be the resulting graph.

Let $i_0 = \log N$. For $i=i_0, \ldots, h-1$, and $j=1,\ldots, N$, let
\begin{equation*}
    \Shift(i,j) =  1+(j-1)2^i/ N - 2^i.
\end{equation*}
For a fixed $i$, the $\Shift(i,j)$s are $N$ equally spaced numbers in
the interval $\IRY{1-2^{i}}{1-2^i/N}$, starting at its left
endpoint. Here, $i$ is the \emphi{resolution} of $\Shift(i,j)$, the
\emphi{shift} corresponding to resolution $i$ is $2^i / N$, and the
number of different shifts is $N$. For $k=0,\ldots, n/2^i$, and $i,j$
as above, the corresponding \emphi{block} is
\begin{equation*}
    \IZ{i}{j}{k} =
    \IRY{\Shift(i,j)+k2^i}{\Shift(i,j)+(k+1)2^i -1}.
\end{equation*}
Such a block is an interval of length $2^i$ that starts at
$\Shift(i,j)+k2^i$, see \figref{shifts-example}. The set of all
intervals / blocks of interest is
\begin{figure}[t]
    \begin{center}
        \includegraphics[width=0.95\linewidth]%
        {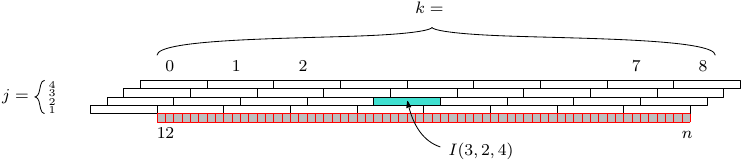}%
        \caption{The shifted intervals $I(i,\cdot,\cdot)$ for $i=3$
           with $\nz=4$ and $n=64$. Each interval has length $2^i=8$,
           there are $\nz=4$ different shifts and there are
           ${n}/{2^i}+1=9$ blocks per each shift.}
        \figlab{shifts-example}
    \end{center}
\end{figure}
\begin{equation}
    \ISet = \Set{ \IZ{i}{j}{k}}{%
       \begin{array}{c}
         i=i_0, \ldots, \log n\\
         j=1,\ldots, \nz\\
         k=0,\ldots, n/2^i
       \end{array}%
    }.
    \eqlab{i-set}
\end{equation}

\paragraph{Constructing the graph $\Geps$.} %
Let $\ExpZ{i}{j}{k}$ denote the expander graph of \lemref{expander},
constructed over $\IZ{i}{j}{k}$ and $\IZ{i}{j}{k+1}$, with the value
of the parameter $\epsA$ as specified in \Eqref{n-z-val}.  We define
$\Geps$ to be the union of all the graphs $\ExpZC$ over all choices of
$i,j,k$, and also including the graph $\Graph_0$ (described
above). The last step is to delete vertices from $\Geps$ that are
outside the range of interest $\IRX{n}$.

\begin{remark*}
    As before, if $n$ is not a power of two, repeat the construction on
    $\IRX{\powTwoX{n}}$, and remove redundant vertices.
\end{remark*}

\subsubsection{Analysis of \texorpdfstring{$\Geps$}{G epsilon}}

\begin{lemma}
    \lemlab{num_edges}%
    The graph $\Geps$ has $\Of( \epsR^{-6} n \log n)$ edges.
\end{lemma}
\begin{proof}
    There are $\Of(\log n)$ resolutions. For every such resolution the
    number of different shifts is $\nz = \Of( 1/\epsR^2) $. For every
    shift, the number of edges created is bounded by
    $\Of(n \epsA^{-2}) = \Of(n /\epsR^4)$, by \lemref{expander}. Thus,
    $\Geps$ has $\Of( \epsR^{-6} n \log n)$ edges.~
\end{proof}

In the following, let $\ILY{s}{\ell} = \IRY{s}{s + \ell -1}$ be the
set of consecutive integers starting at $s$ containing $\ell$ numbers.

\begin{definition}
    For two vertices $x,y \in \IRX{n}$, $y$ is a \emphi{descendant} of
    $x$ (and $x$ is an \emphi{ancestor} of $y$) in $\Geps$, if
    $x < y$ and there is a $1$-path between $x$ and $y$ in
    $\Geps$. For a set $\BSet \subseteq \IRX{n}$, and a vertex $s$,
    let $\DesSet = \DesSet(\Geps,s,\BSet)$ be the set of all
    descendants of $s$ in the graph $\Geps \setminus \BSet$. Similarly, for a
    vertex $t$, let $\AnsSet = \AnsSet(\Geps,t,\BSet)$ be the set of
    ancestors of $t$ in $\Geps \setminus \BSet$.

    For an interval $I \subseteq \IRX{n}$, the set $I\cap \DesSet$ is
    the set of all nodes in $I$ that are descendants of $s$ in the
    graph $\Geps \setminus \BSet$. In a symmetric fashion, the set of
    ancestors in $I$ that can reach a node $t$ is denoted by
    $I\cap \AnsSet$.
\end{definition}

Now, we show that if a point outside of the shadow has a reasonably large fraction of descendants in an interval, then, one can find an extended interval, which is $\Theta(1/\epsR)$ times longer and has the same property. The crucial part is to carefully choose two consecutive blocks of $\ISet$, such that the expander between them can be used to extend the set of descendants.

\begin{lemma}
    \lemlab{expand_r}%
    Let $\BSet \subseteq \IRX{n}$ be the set of deleted locations,
    $\alpha = 1-\epsR/4$ and $\sLoc$ be a location in $\IRX{n}$ that
    is not in the $\alpha$-shadow of $\BSet$. Let $\lShort \geq \nz$
    be an integer number, and let $c\geq 512$ be the constant from the
    construction.  Let $\DesSet = \DesSet(\Geps,\sLoc,\BSet)$, and
    assume that
    \begin{math}
        \cardin{\DesX{\ILY{\sLoc}{\lShort}\bigr.}}%
        \geq%
        (\epsR/32)\lShort.
    \end{math}
    Then, for some number $\lLong$,
    $8\lShort/\epsR \leq \lLong \leq (c/8) \lShort/\epsR$, we have
    \begin{math}
        \cardin{\DesX{\ILY{\sLoc}{\lLong}\bigr.}}%
        \geq (\epsR/32)\lLong.
    \end{math}
\end{lemma}

\SaveIndent%
\begin{proof}
    The idea is to choose the right resolution in the construction of
    $\Geps$. As a first step, let
    \begin{equation*}
        \Shift%
        =%
        \powTwoX{\epsR \lShort / 64}
        \qquad \implies \qquad
        \epsR \lShort/64 \leq \Shift \leq \epsR \lShort/32%
        \eqlab{S}%
    \end{equation*}
    be the desired shift. We pick the resolution $i$ such that the
    shift used $2^i/\nz$ is equal to $\Shift$ (i.e.,
    $\Shift= 2^i/\nz$).  This implies that $i = \log (\nz
    \Shift)$. There is a choice of $j$ and $k$, such that the right
    endpoint of $\ILeft = \IZ{i}{j}{k}$ lies in the interval
    $\ILY{\sLoc+h}{\Shift}$. Notice that
    $\ILY{\sLoc}{\lShort} \subseteq \ILeft$, since
    \begin{equation*}
        \lShort + \Shift \leq \pth{1+{64}/{\epsR}} \Shift =
        \pth{1+\frac{64}{\epsR}} \frac{2^i}{\nz} \leq
        \pth{1+\frac{64}{\epsR}} \frac{\epsR^2}{c} 2^i \leq 2^i
    \end{equation*}
    holds.  Let $\IRight = \IZ{i}{j}{k+1}$ and
    $\lLong = \rightX{\IRight} - \sLoc +1$, where $\rightX{\IRight}$
    is the right endpoint of the interval $\IRight$, see
    \figref{shift-resolution}.  Observe that
    $\epsR h /64 \leq \Shift \leq \epsR h/32$ and
    \begin{equation*}
        \lLong%
        \geq%
        2^i%
        =%
        \nz \Shift%
        \geq%
        \frac{c}{\epsR^2} \cdot \frac{\epsR \lShort}{64}%
        =%
        \frac{c}{64}\cdot \frac{\lShort}{\epsR}%
        \geq%
        \frac{8\lShort}{\epsR},
    \end{equation*}

    \begin{figure}[t]
        \centering \includegraphics{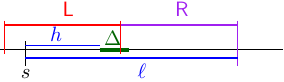}%
        \caption{The intervals $\ILeft$ and $\IRight$ and their relation to
           $\sLoc,\lShort,\Shift$ and $\lLong$.}
        \figlab{shift-resolution}
    \end{figure}

    \noindent%
    since $c\geq 512$.  Similarly,
    \begin{equation*}
        \lLong%
        \leq%
        2\cdot 2^i%
        =%
        2 \nz \Shift%
        \leq%
        2 \cdot \frac{2c}{\epsR^2} \cdot \frac{\epsR \lShort}{32}%
        =%
        \frac{c}{8}\cdot \frac{\lShort}{\epsR}.
    \end{equation*}

    Let $U = \ILY{\sLoc}{ \lShort }\cap\DesSet$. By assumption,
    $\cardin{U} \geq (\epsR/32) h$.  Since the interval $\ILeft$ is of
    length $2^i$, we have
    \begin{equation*}
        \frac{\cardin{\ILeft \cap \DesSet}}{\cardin{\ILeft}}%
        \geq %
        \frac{\cardin{U}}{2^i}%
        \geq%
        \frac{(\epsR/32) \lShort}{\nz \Shift}
        \geq%
        \frac{(\epsR/32) \lShort}{\nz (\epsR/32) \lShort}%
        =%
        \frac{1}{\nz}%
        \geq%
        \epsA.
    \end{equation*}
    Since $\sLoc$ is not in the $\alpha$-shadow of $\BSet$, it follows
    that the interval $\ILY{\sLoc}{\lLong}$ contains at least
    $(\epsR/4) \lLong$ elements that are not in $\BSet$.  Let $\tau$
    be the fraction of elements of $\IRight$ that are not in $\BSet$.
    We have that
    \begin{align*}
      \tau%
      =%
      \frac{\cardin{\IRight \setminus \BSet}}{\cardin{\IRight}}%
      & \geq%
        \frac{(\epsR/4)\lLong - \lShort - \Shift}{2^i} %
        \geq %
        \frac{(\epsR/4)(2^i + h) - (\lShort + \Shift)}{2^i} %
      \\&%
      \geq %
      \frac{(\epsR/4)
      \pth{\bigl.2^i + ({32}/{\epsR \nz})2^i} -
      ({64}/{\epsR} + 1){2^i}/{\nz}}{2^i} \\%
      & =%
        \frac{\epsR}{4} + \frac{8}{\nz} - \pth{1+
        \frac{64}{\epsR}}\frac{1}{\nz} %
        \geq %
        \frac{\epsR}{4} - \frac{64}{\epsR \nz} %
        \geq %
        \frac{\epsR}{4} - \frac{64\epsR}{c} %
        \geq %
        \frac{\epsR}{8}.
    \end{align*}
    Let $U' \subseteq \IRight$ be the set of all nodes that are
    connected by an edge of $\Geps$ to $U$. Note, that all the nodes
    of $U'$ are descendants of $s$. The graph $\ExpZ{i}{j}{k}$
    guarantees that $\cardin{U'} \geq (1-\epsA) \cardin{\IRight}$,
    where $\ExpZ{i}{j}{k}$ is the expander graph built over $\ILeft$
    and $\IRight$. We have that
    \begin{align*}
      \cardin{\bigl.\ILY{\sLoc}{\lLong}\cap\DesSet}%
      & \geq %
        \cardin{ (\IRight \setminus \BSet) \cap U'}%
        = %
        \cardin{ \IRight \setminus \BSet}
        - \cardin{(\IRight \setminus \BSet ) \cap \overline{U'}} \\
      & \geq %
        \cardin{ \IRight \setminus \BSet} - \cardin{\IRight \cap \overline{U'}}
        \geq%
        \cardin{ \IRight \setminus \BSet} - \epsA 2^i%
        =%
        (\tau - \epsA) 2^i.
    \end{align*}
    Since $\epsA \leq \epsR/16$, we have
    \begin{math}
        \displaystyle%
        \frac{ \cardin{\ILY{\sLoc}{\lLong}\cap\DesSet}}{\lLong} %
        \geq%
        \frac{ (\tau - \epsA) 2^i}{2\cdot2^i}%
        =%
        \frac{\tau - \epsA}{2}%
        \geq%
        \frac{\epsR/8 - \epsR/16}{2}%
        =%
        \frac{\epsR}{32}.
    \end{math} 
\end{proof}

\begin{remark*}
    \remlab{symmetric}%
    One can state a symmetric version of \lemref{expand_r} about the
    number of ancestors that can reach a target node $\tLoc$.
\end{remark*}

\begin{lemma}
    \lemlab{s_path}%
    Let $\BSet \subseteq \IRX{n}$ be the set of faulty vertices, and
    let $\ShadowY{\alpha}{\BSet}$ be its $\alpha$-shadow with
    $\alpha = 1-\epsR/4$. Let $\sLoc, \tLoc$ be two vertices in
    $\IRX{n} \setminus \ShadowY{\alpha}{\BSet}$, such that
    $\sLoc < \tLoc$.  Then, there is a $1$-path between $\sLoc$ and
    $\tLoc$ in $\Geps \setminus \BSet$. Further, this path between
    $\sLoc$ and $\tLoc$ uses at most $2\log{n}$ edges.
\end{lemma}
\begin{proof}
    If $\cardin{\sLoc - \tLoc} \leq 3N$, then the two vertices are
    connected by an edge in $\Geps$ by construction, and the claim
    holds.

    Let $\ILeft$ and $\IRight$ be two adjacent consecutive blocks of
    the same size in $\ISet$ (see \Eqref{i-set}), such that
    $\sLoc \in \ILeft$ and $\tLoc \in \IRight$, and these are the
    smallest blocks for which this property holds. If there are
    several pairs of intervals of the same size that have the desired
    property, we pick the pair such that
    $\min\pth{\rightX{\ILeft} - \sLoc, \, \tLoc - \leftX{\IRight}}$ is
    maximized (i.e., the common boundary between the two intervals is
    as close to the middle $(\sLoc+\tLoc)/2$ as possible).  Let
    $2^i = \cardin{\IRight} = \cardin{\ILeft}$. It is easy to verify
    that $2^i/2 \leq \cardin{\IRY{\sLoc}{\tLoc}} \leq 2\cdot
    2^i$. Indeed, the lower bound holds by the minimality of $\ILeft$
    and $\IRight$. Otherwise, the right half of $\ILeft$ and the left
    half of $\IRight$ would also be a valid choice and would have
    smaller size. The upper bound follows from the fact that
    $\cardin{\ILeft}+\cardin{\IRight} = 2\cdot 2^i$.

    Set $L_0 = \ILY{\sLoc}{\nz}$ and
    $R_0 = \IRY{\tLoc - \nz +1}{\tLoc}$. Since $\sLoc$ and $\tLoc$ are
    not in the $\alpha$-shadow, we have that
    $\cardin{L_0 \setminus \BSet} \geq (\epsR/4) \cardin{L_0}$ and
    $\cardin{R_0 \setminus \BSet} \geq (\epsR/4) \cardin{R_0}$.  For
    $i > 0$, in the $i$\th iteration, let $L_i$ be the interval
    starting at $s$ of length $\Theta( \cardin{L_{i-1}} /\epsR)$ such
    that a fraction of at least $\epsR/32$ of its elements are descendants
    of $s$ that are not in $\BSet$. The existence of such an interval
    is guaranteed by \lemref{expand_r}. Similarly, we expand the right
    interval $R_{i-1}$ in a symmetric way.

    Let $j$ be the first iteration such that
    $L_{j+1} \not \subseteq \ILeft$. By the choice of $\ILeft$ and
    $\IRight$ and by \lemref{expand_r}, we have
    \begin{equation*}
        \frac{2^i}{4}-\frac{2^i}{\nz}%
        \leq %
        \cardin{L_{j+1}}%
        \leq%
        \frac{c}{8\epsR} \cardin{L_{j}}.
    \end{equation*}
    This implies that
    \begin{align*}
      \frac{\cardin{\ILeft\cap\DesSet}}{\cardin{\ILeft}}%
      &\geq%
        \frac{\cardin{L_j\cap\DesSet}}{\cardin{\ILeft}}%
        \geq%
        \frac{ (\epsR/32)\cardin{L_j} }{ 2^i } %
        \geq
        \frac{\epsR}{32} \cdot \frac{8\epsR}{c}
        \pth{ \frac{1}{4} - \frac{1}{\nz} }
        \geq%
        \frac{\epsR}{32} \cdot \frac{8\epsR}{c}\cdot \frac{1}{8}
      \geq \frac{1}{32\nz}%
      =%
      \epsA.
    \end{align*}
    Applying the same argumentation, using \lemref{expand_r}~for the
    reachable ancestors, we have that
    \begin{equation*}
        \cardin{\IRight \cap \AnsSet} / \cardin{\IRight}%
        \geq%
        \epsA
    \end{equation*}
    (i.e., there are at least $\epsA \cardin{\IRight}$ elements in
    $\IRight$ that have a $1$-path to $\tLoc$ in
    $\Geps \setminus \BSet$). The graph $\Geps$ contains an expander
    $\ExpZ{i}{j}{k}$ built over $\ILeft$ and $\IRight$. By the
    pigeonhole principle and the properties of the expander between
    $\ILeft$ and $\IRight$, there is an edge between a vertex of
    $\ILeft \cap \DesSet$ and a vertex of $\IRight \cap \AnsSet$. That
    is, there is a $1$-path between $s$ and $t$ in
    $\Geps \setminus \BSet$, as desired.

    By \lemref{expand_r} we have
    $8\cardin{L_{i}} \leq ({8}/{\epsR})\cardin{L_{i}} \leq
    \cardin{L_{i+1}}$ for $i=0,\dots,j$. Therefore, the number of
    iterations we do to expand $L_0$ is less than $\log{n}$. The same
    is true for $R_0$. Thus, the number of edges that we used for the
    $1$-path is bounded by $2\log{n}$. 
\end{proof}

\begin{theorem}
    \thmlab{useless_eps}%
    For parameters $n$ and $\epsR >0$, the graph $\Geps$ constructed
    over $\IRX{n}$, is a $\epsR$-reliable exact spanner.  Furthermore,
    $\Geps$ has $\Of(\epsR^{-6} n \log n)$ edges.
\end{theorem}
\begin{proof}
    The bound on the number of edges is from \lemref{num_edges}.

    Next, fix the set $\BSet$.  Define the set $\EBSet$ to be the
    $(1-\epsR/4)$-shadow of $\BSet$. By \lemref{shadow_eps} we have
    that
    \begin{math}
        \cardin{\EBSet}%
        \leq%
        \cardin{\BSet}/(2(1-\epsR/4)-1)%
        =%
        \cardin{\BSet}/(1 - \epsR/2 )%
        \leq%
        (1+\epsR)\cardin{\BSet}.
    \end{math}

    A $1$-path in $\Geps \setminus \BSet$ between any two vertices in
    $\IRX{n} \setminus \EBSet$ exists by \lemref{s_path}. 
\end{proof}

\section{Building a reliable spanner in \texorpdfstring{$\Re^d$}{Rd}}
\seclab{r:d:general}

\subsection{A first construction}

In the following, we assume that $\PS \subseteq [0,1)^d$ -- this can
be done by an appropriate scaling and translation of space. We use a
recent result of Chan \etal \cite{chj-lsota-18}, that introduced
\emph{locality-sensitive orderings}.  These orderings (which are the
same as total orders and are extensions of the $\mathcal{Z}$-order)
can be thought as an alternative to quadtrees and related structures.
For an ordering $\order$ of $[0,1)^d$, and two points
$\pp,\pq \in [0,1)^d$, such that $\pp \prec_\order \pq$, let
\begin{equation*}
    (\pp,\pq)_{\order} = \Set{\pz \in [0,1)^d}{ \pp \prec_\order \pz
       \prec_\order \pq}   
\end{equation*}
be the open interval between $\pp$ and $\pq$ in the order $\order$.
Further, let $\ballY{\pp}{r} = \Set{\pz \in [0,1)^d}{\dY{\pp}{\pz} \leq r}$ denote the ball centered at $\pp$ with radius $r$.

\begin{theorem}[\cite{chj-lsota-18}]
    \thmlab{lso}%
    For $\epsB \in (0,1)$, there is a set $\ordAll(\epsB)$ of at most
    $M(\epsB) = \Of( \epsB^{-d} \log \epsB^{-1} )$ orderings of
    $[0,1)^d$, such that for any two (distinct) points
    $\pp, \pq \in [0,1)^d$, with $\ell = \dY{\pp}{\pq}$, there is an
    ordering $\order \in \ordAll$, and a point $\pz \in [0,1)^d$, such
    that %
    \medskip%
    \begin{compactenumi}
        \item $\pp \prec_\order \pq$,

        \smallskip%
        \item
        $(\pp,\pz)_\order \subseteq \ballY{\pp}{\bigl.  \epsB \ell}$,

        \smallskip%
        \item
        $(\pz,\pq)_\order \subseteq \ballY{\pq}{\bigl.  \epsB \ell}$,
        and

        \smallskip%
        \item $\pz \in \ballY{\pp}{\bigl.  \epsB \ell}$ or
        $\pz \in \ballY{\pq}{\bigl.  \epsB \ell}$.
    \end{compactenumi}
    \medskip%
    Furthermore, given such an ordering $\order$, and two points
    $\pp,\pq$, one can compute their ordering, according to $\order$,
    using $\Of(d\log \epsB^{-1})$ arithmetic and bitwise-logical
    operations.
\end{theorem}

First, we give a very simple construction and analysis, for building
reliable spanners, using the theorem above and the one-dimensional
construction. We present it to convey the basic principle of this
technique. Then, by tuning the parameters, we repeat the construction
to obtain a reliable spanner of size
$\Of\pth{n \log n (\log \log n)^{6}}$. This construction has a more
elaborate analysis, with a similar ideas but used in an iterative
manner.

\subsubsection{Construction in detail}
Given a set $\PS$ of $n$ points in $[0,1)^d$, and parameters
$\eps, \epsR \in (0,1)$, let $\epsB = \eps / (\constA \log n)$,
\begin{equation*}
    M%
    =%
    M(\epsB)%
    =%
    \Of(\epsB^{-d} \log \epsB^{-1} )%
    =%
    \Of\Bigl( \eps^{-d} \log^d n \log \frac{\log n}{\eps} \Bigr),
\end{equation*}
and $\constA$ be some sufficiently large constant. Next, let
$\epsR' = \epsR / M$, and let $\ordAll = \ordAll(\epsB)$ be the set of
orderings of \thmref{lso}. For each ordering $\order \in \ordAll$,
compute the $\epsR'$-reliable exact spanner $\Graph_\order$ of $\PS$,
see \thmref{useless_eps}, according to $\order$. Let $\Graph$ be the
resulting graph by taking the union of $\Graph_\order$ for all
$\order \in \ordAll$.

\subsubsection{Analysis}

\begin{lemma}
    \lemlab{inferior}%
    The graph $\Graph$, constructed above, is a $\epsR$-reliable
    $(1+\eps)$-spanner and has size
    \begin{equation*}
        \Of\Bigl( \eps^{-7d} \epsR^{-6} n \log^{7d} n \log^{7} \frac{\log
           n}{\eps}\Bigr).
    \end{equation*}
\end{lemma}
\begin{proof}
    Given a (failure) set $\BSet \subseteq \PS$, let $\EBSet$ be the
    union of all the harmed sets resulting from $\BSet$ in
    $\Graph_{\order}$, for all $\order \in \ordAll$.  We have that
    $\cardin{\EBSet} \leq (1 + M \cdot \epsR') \cardin{\BSet} = (1 +
    \epsR) \cardin{\BSet}$.

    Consider any two points $\pp, \pq \in \PS \setminus \EBSet$. By
    \thmref{lso}, for $\ell = \dY{\pp}{\pq}$, there exists an ordering
    $\order \in \ordAll$, and a point $\pz \in [0,1)^d$, such that
    $(\pp,\pz)_\order \subseteq \ballY{\pp}{\epsB \ell}$ and
    $(\pz,\pq)_\order \subseteq \ballY{\pq}{\epsB \ell}$ (and $\pz$ is
    in one of these balls).

    By \lemref{s_path}, the graph
    $\Graph_\order \setminus \BSet \subseteq \Graph \setminus \BSet$
    contains a monotone path $\pi$, according to $\order$, with
    $h = \Of( \log n)$ hops, connecting $\pp$ to $\pq$. Let
    $\pp = \pp_1, \ldots,\pp_{h+1} = \pq$ be this path. Observe that
    there is a unique index $i$, such that
    $\pz \in (\pp_i,\pp_{i+1})$. We have the following: \smallskip%
    \begin{compactenumA}
        \item $\forall j \neq i$\quad
        $\dY{\pp_j}{\pp_{j+1}} \leq 2\epsB \ell$, since $\pp_j$ and $\pp_{j+1}$ are contained in a ball of radius $\epsB \ell$.

        \smallskip%
        \item
        $\dY{\pp_i}{\pp_{i+1}} \leq \dY{\pp_i}{\pp} + \dY{\pp}{\pq} +
        \dY{\pq}{\pp_{i+1}} \leq \ell + 2\epsB\ell$.
    \end{compactenumA}
    \smallskip%
    As such, the total length of $\pi$ is
    $\sum_{j=1}^h \dY{\pp_j}{\pp_{j+1}} = (1+ 2\epsB h)\ell \leq
    (1+\eps)\ell$, as desired, if $\constA$ is sufficiently large.
    Namely, $\Graph$ is the desired reliable spanner.

    The number of edges of $\Graph$ is
    \begin{equation*}
      M \cdot \Of\pth{ (\epsR')^{-6} n \log n}
      =%
      \Of\pth{ M (M/\epsR )^6 n \log n}%
      =%
      \Of\Bigl( \eps^{-7d} \epsR^{-6}  n \log^{7d} n \log^{7}
      \frac{\log n}{\eps} \Bigr).        
    \end{equation*}
\end{proof}

\subsection{An improved construction}
Given a set $\PS$ of $n$ points in $[0,1)^d$, and parameters
$\eps, \epsR \in (0,1)$, let $\epsB = \eps / \constA $,
\begin{equation*}
    M%
    =%
    M(\epsB)%
    =%
    \Of(\epsB^{-d} \log \epsB^{-1} )%
    =%
    \Of\bigl( \eps^{-d} \log \eps^{-1}\bigr),
\end{equation*}
and $\constA$ be some sufficiently large constant. Next, let
$\epsR' = \epsR / (3N \cdot M)$ where $N=\ceil{\log\log n}+1$, and let
$\ordAll = \ordAll(\epsB)$ be the set of orderings of
\thmref{lso}. For each ordering $\order \in \ordAll$, compute the
$\epsR'$-reliable exact spanner $\Graph_\order$ of $\PS$, see
\thmref{useless_eps}, according to $\order$. Let $\Graph$ be the
resulting graph by taking the union of $\Graph_\order$ for all
$\order \in \ordAll$.

\begin{theorem}
    \thmlab{main}%
    The graph $\Graph$, constructed above, is a $\epsR$-reliable
    $(1+\eps)$-spanner and has size
    \begin{equation*}
        \Of\Bigl( \eps^{-7d} \log^7 \frac{1}{\eps} \cdot \epsR^{-6}
        n \log n (\log \log n)^{6}\Bigr).
    \end{equation*}
\end{theorem}
\begin{proof}
    First, we show the bound on the size. There are $M$ different
    orderings for which we build the graph of
    \thmref{useless_eps}. Each of these graphs has
    $\Of\pth{ (\epsR')^{-6} n \log n}$ edges. Therefore, the size of
    $\Graph$ is
    \begin{align*}
      M \cdot \Of\pth{ (\epsR')^{-6} n \log n}
      &=%
        \Of\Bigl( M \Bigl( \frac{3N \ts M}{ \epsR } \Bigr)^6 n \log n\Bigr)%
      =%
        \Of\Bigl( \eps^{-7d} \log^7 \frac{1}{\eps} \cdot \epsR^{-6}
        n \log n (\log \log n)^{6}\Bigr).
    \end{align*}
    Next, we identify the set of harmed vertices $\EBSet$ given a set
    of failed vertices $\BSet \subseteq \PS$. First, let $\BSet_1$ be
    the union of all the $(1-\epsR'/4)$-shadows resulting from $\BSet$
    in $\Graph_{\order}$, for all $\order \in \ordAll$. Then, for
    $i = 2,\dots,N$, we define $\BSet_i$ in a recursive manner to be
    the union of all the $(1-\epsR'/4)$-shadows resulting from
    $\BSet_{i-1}$ in $\Graph_{\order}$, for all $\order \in
    \ordAll$. We set $\EBSet = \BSet_N$.

    By the recursion and \lemref{shadow_eps} we have that
    \begin{align*}
      \cardin{\BSet_i}%
      &\leq
        \Bigl(\frac{\cardin{\BSet_{i-1}}}{(2(1-\epsR'/4)-1)} -
        \cardin{\BSet_{i-1}}\Bigr)M + \cardin{\BSet_{i-1}}%
      \\&%
      =%
        \frac{\cardin{\BSet_{i-1}} -
        (1-\epsR'/2)\cardin{\BSet_{i-1}}}{(1 -
        \epsR'/2 )}M + \cardin{\BSet_{i-1}}%
      = %
      \frac{\epsR'\cardin{\BSet_{i-1}}}{(2 - \epsR')}M + \cardin{\BSet_{i-1}}%
      \\&%
      \leq%
      (1+\epsR'M)\cardin{\BSet_{i-1}}%
      =%
      \left(1+\frac{\epsR}{3N}\right)\cardin{\BSet_{i-1}}.
    \end{align*}
    Therefore, we obtain
    \begin{equation*}
        \cardin{\EBSet} = \cardin{\BSet_N}%
        \leq%
        \pth{1+\frac{\epsR}{3N}}^N  \cardin{\BSet}%
        \leq%
        \exp \pth{N \frac{\epsR}{3N}}  \cardin{\BSet}%
        \leq%
        (1 + \epsR) \cardin{\BSet},
    \end{equation*}
    using $1+x \leq e^x \leq 1+3x$, for $x\in[0,1]$.

    The claim is that there is a $(1+\eps)$-path $\pout$ between any
    pair of vertices
    $\pp,\pq\in \PS\setminus \EBSet \equiv \PS\setminus\BSet_N$ such
    that the path $\pout$ does not use any vertices of $\BSet$.  By
    \thmref{useless_eps}, for the right choice of $\order$, the graph
    $\Graph_\order \setminus \BSet_{N-1} \subseteq \Graph \setminus
    \BSet_{N-1}$ contains a monotone path connecting $\pp$ to $\pq$,
    according to $\order$. Observe that there is a unique edge
    $(\pp',\pq')$ on this path that ``jumps'' from the locality of
    $\pp$ to the locality of $\pq$. Formally, we have the following: %
    \medskip%
    \begin{compactenumA}
        \item
        \begin{math}
            \dY{\pp'}{\pq'}%
            \leq%
            \dY{\pp}{\pq} + 2\epsB \dY{\pp}{\pq} =%
            (1 + 2 \eps/\constA) \dY{\pp}{\pq}.
        \end{math}

        \smallskip%
        \item
        $\dY{\pp}{\pp'} \leq 2\epsB\dY{\pp}{\pq} =
        2(\eps/\constA)\dY{\pp}{\pq} $, and similarly

        $\dY{\pq}{\pq'} \leq 2(\eps/\constA)\dY{\pp}{\pq} $.

        \smallskip%
        \item $\pp',\pq'\in \PS\setminus\BSet_{N-1}$.
    \end{compactenumA}
    \medskip%
    We fix the edge $(\pp',\pq')$ to be used in the computed path
    $\pout$ connecting $\pp$ to $\pq$. We still need to build the two
    parts of the path $\pout$ between $\pp,\pp'$ and $\pq,\pq'$.

    This procedure reduced the problem of computing a reliable path
    between two points $\pp,\pq \in \PS\setminus\BSet_N$, into
    computing two such paths between two points of
    $\PS\setminus\BSet_{N-1}$ (i.e., $\pp, \pp'$ and $\pq, \pq'$). The
    benefit here is that both $\dY{\pp}{\pp'}$ and $\dY{\pq}{\pq'}$
    are much smaller than $\dY{\pp}{\pq}$.  We now repeat this
    refinement process $N-1$ times.

    To this end, let $Q_i$ be the set of active pairs that needs to be
    connected in the $i$\th level of the recursion. Thus, we have that
    $Q_0 = \{(\pp,\pq) \},\; Q_1=\{(\pp,\pp'), (\pq,\pq')\}$, and so
    on. We repeat this $N-1$ times. In the $i$\th level there are
    $\cardin{Q_i}=2^i$ active pairs. Let $(x,y)\in Q_i$ be such a
    pair. Then, there is an edge $(x',y')$ in the graph
    $\Graph\setminus \BSet_{N-(i+1)}$, such that we have the
    following: %
    \medskip%
    \begin{compactenumA}
        \item
        \begin{math}
            \dY{x'}{y'} \leq \dY{x}{y} (1 + 2 \eps/\constA) \leq (2
            \eps/\constA)^i (1 + 2 \eps/\constA) \dY{\pp}{\pq}.
        \end{math}

        \smallskip%
        \item
        $\dY{x}{x'} \leq 2(\eps/\constA)\dY{x}{y} \leq
        (2\eps/\constA)^{i+1}\dY{\pp}{\pq} $, and
        
        $\dY{y}{y'} \leq (2\eps/\constA)^{i+1}\dY{\pp}{\pq} $.

        \smallskip%
        \item $x',y'\in \PS\setminus\BSet_{N-(i+1)}$.
    \end{compactenumA}
    \medskip%
    The edge $(x',y')$ is added to the path $\pout$. After $N-1$
    iterations the set of active pairs is $Q_{N-1}$ and for each pair
    $(x,y)\in Q_{N-1}$ we have that $x,y\in \PS\setminus\BSet_1$. For
    each of these pairs $(x,y)\in Q_{N-1}$ we apply \thmref{lso} and
    \thmref{useless_eps} to obtain a path of length at most
    $\dY{x}{y} 2\log n$ between $x$ and $y$ (and this subpath of
    course does not contain any vertex in $\BSet$). We add all these
    subpaths to connect the active pairs in the path $\pout$, which
    completes $\pout$ into a path.

    Now, we bound the length of path $\pout$. Since, for all
    $(x,y)\in Q_{N-1}$, we have
    $\dY{x}{y} \leq \dY{p}{q} \cdot (2\eps/\constA)^{N-1}$ and
    $\cardin{Q_{N-1}} = 2^{N-1}$, the total length of the subpaths
    calculated, in the last step, is
    \begin{align*}
      &
        \sum_{(x,y) \in Q_{N-1}} \textrm{length}\pth{ \bigl. \pout[ x, y] }
        \leq%
        2^{N-1} \dY{p}{q} \cdot \Bigl(\frac{2\eps}{\constA}\Bigr)^{N-1} 2\log n
      \\&%
      \qquad%
      \leq %
      \dY{p}{q} \cdot \Bigl(\frac{4\eps}{\constA}\Bigr)^{\log\log n} 2\log n
      \leq%
      \dY{p}{q} \cdot \eps^{\log\log n} \Bigl(\frac{4}{\constA}\Bigr)^{\log\log n}
      2\log n
      \\&%
      \qquad
      \leq %
      \dY{p}{q} \cdot \frac{\eps}{4}\cdot \frac{1}{\log n}\cdot 2\log n
      =%
      \frac{\eps}{2}\dY{p}{q} ,
    \end{align*}
    for large enough $n$ and $\constA\geq 8$.  The total length of the
    long edges added to $\pout$ in the recursion, is bounded by
    \begin{align*}
      &
        \sum_{i=0}^{N-2} 2^i \dY{p}{q} \Bigl(\frac{2\eps}{\constA}\Bigr)^i
        \Bigl(1 + \frac{2\eps}{\constA}\Bigr)
        \leq %
        \dY{p}{q} \Bigl(1 + \frac{2\eps}{\constA}\Bigr) \sum_{i=0}^{\infty}
        \Bigl(\frac{4\eps}{\constA}\Bigr)^i
      \\&\qquad\qquad%
      = %
      \dY{p}{q} \Bigl(1 + \frac{2\eps}{\constA}\Bigr)
      \frac{1}{1-{4\eps}/{\constA}}
      =%
      \dY{p}{q} \Bigl(1 + \frac{6\eps}{\constA - 4\eps}\Bigr)
      \\&\qquad\qquad%
      \leq%
      \pth{1 + \frac{\eps}{2}}\dY{p}{q},
    \end{align*}
    which holds for $\constA \geq 16$. Therefore, the computed path
    $\pout$ between $\pp$ and $\pq$ is a $(1+\eps)$-path in
    $\Graph\setminus\BSet$, which concludes the proof of the theorem.
\end{proof}

\section{Construction for points with bounded %
   spread in \texorpdfstring{$\Re^d$}{Rd}}
\seclab{r:d:bounded:spread}

The input is again a set $\PS\subset \Re^d$ of $n$ points, and
parameters $\epsR \in (0,1/2)$ and $\eps \in (0,1)$. The goal is to
build a $\epsR$-reliable $(1+\eps)$-spanner on $\PS$ that has optimal
size under the condition that the spread $\Phi(\PS)$ is bounded by a
polynomial of $n$.

\subsection{Preliminaries}

\begin{definition}
    For a point set $\PS \subseteq \Re^d$, let
    $\diamX{\PS}= \max_{\pp,\pq\in \PS} \dY{\pp}{\pq}$ denote the
    \emphi{diameter} of $\PS$.  Let
    \begin{math}
        \CPX{\PS} = \min_{\pp,\pq \in \PS, \pp \neq \pq} \dY{\pp}{\pq}
    \end{math}
    denote the \emphi{closest pair} distance in $\PS$.  Furthermore,
    let
    \begin{math}
        \SpreadC(\PS) = \diamX{\PS}/\CPX{\PS}%
    \end{math}
    be the \emphi{spread} of $\PS$.
\end{definition}

\begin{definition}
    Let $s>0$ be a real number and let $\PB$ and $\PC$ be sets of
    points in $\Re^d$. The sets $\PB$ and $\PC$ are
    \emphi{$s$-separated} if
    $\distSetY{\PB}{\PC} \geq s \cdot
    \max\pth{\diamX{\PB},\diamX{\PC}}$, where
    $\distSetY{\PB}{\PC} = \min_{\pp \in \PB, \pq \in \PC}
    \dY{\pp}{\pq}$.
\end{definition}

\begin{definition}
    Let $\PS$ be a set of $n$ points in the plane and let $s>0$ be a
    real number. An \emph{$s$-well-separated pair decomposition}
    (\emphi{$s$-\WSPD}) of $\PS$ is a collection
    $\WS = \{(\PB_i,\PC_i)\}_{i=1}^m$ of pairs of subsets of $\PS$,
    such that:
    \begin{compactenumI}
        \smallskip
        \item The sets $\PB_i$ and $\PC_i$ are $s$-separated, for all
        $i=1,2,\dots,m$, and \smallskip%

        \item For any $\pp,\pq \in \PS$ ($\pp\neq \pq$) there exists a
        unique pair $(\PB_i,\PC_i) \in \WS$), such that
        $\pp \in \PB_i$ and $\pq\in \PC_i$ (or $\pq\in \PB_i$ and
        $\pp\in \PC_i$).
    \end{compactenumI}
\end{definition}

The well-separated pairs decomposition was introduced by Callahan and
Kosaraju \cite{ck-dmpsa-95}. The \emphi{size} of a \WSPD is the number
of pairs $m$, and the \emphi{weight} of a pair decomposition
$\mathcal{W}$ is defined as
$\omega(\mathcal{W}) = \sum_{i=1}^m \pth{\cardin{\PB_i} +
   \cardin{\PC_i}}$.

There are several ways to compute an $s$-\WSPD. Here, we use a
quadtree-based approach, which has important properties that we can
exploit. More precisely, we use the following result of Abam and
Har-Peled \cite[Lemma~2.8]{ah-ncsa-12} for computing a \WSPD.
\begin{lemma}
    \lemlab{WSPD-comp}%
    Let $\PS$ be a set of $n$ points in $\Re^d$, with spread
    $\Phi=\Phi(\PS)$, and let $\eps>0$ be a parameter. Then, one can
    compute an $\eps^{-1}$-\WSPD for $\PS$ of total weight
    $\Of(n\eps^{-d} \log{\Phi})$. Furthermore, any point of $\PS$
    participates in at most $\Of(\eps^{-d} \log{\Phi})$ pairs.
\end{lemma}

\subsection{The construction of \texorpdfstring{$\GSB$}{G phi}}
\seclab{r:d:bounded:spread:constr} First, compute a quadtree $\Tree$
for the point set $\PS$. For any node $v\in \Tree$, let $\cellX{v}$
denote the \emphi{cell} (i.e.  square or cube, depending on the
dimension) represented by $v$. Let $\PcellX{v} = \cellX{v} \cap \PS$
be the point set stored in the subtree of $v$. Compute a
$({6}/{\eps})$-\WSPD $\WS$ of $\PS$, such that each pair consists of
two cells of $\Tree$, using \lemref{WSPD-comp}. The pairs in $\WS$ can
be represented by pairs of nodes $\{u,v\}$ of the quadtree
$\Tree$. Note that the algorithm of \lemref{WSPD-comp} uses the
diameters and distances of the cells of the quadtree, that is, for a
pair $\{u,v\}\in\WS$, we have
\begin{equation*}
    (6/\eps)\cdot \max\pth{ \bigl.\diamX{\cellX{u}},\diamX{\cellX{v}}}%
    \leq%
    \distSetY{\cellX{u}}{\cellX{v}}.
\end{equation*}

For any pair $\{u,v\}\in \WS$, we build the bipartite expander of
\lemref{expander} on the sets $\PcellX{u}$ and $\PcellX{v}$ such that
the expander property holds with $\epsA=\epsR/8$.  Furthermore, for
every two node $u$ and $v$ that have the same parent in the quadtree
$\Tree$ we add the edges of the bipartite expander of
\lemref{expander} between $\PcellX{u}$ and $\PcellX{v}$.  Let $\GSB$
be the resulting graph when taking the union over all the sets of
edges generated by the above.

\subsection{Analysis}

\begin{lemma}
    \lemlab{G_2d-size}%
    The graph $\GSB$ has
    $\Of\bigl(\epsA^{-2} \eps^{-d}n\log{\Phi(\PS)} \bigr)$ edges.
\end{lemma}
\begin{proof}
    By \lemref{WSPD-comp}, every point participates in
    $\Of(\eps^{-d}\log{\Phi(\PS)})$ \WSPD pairs.  By \lemref{expander}
    the average degree in all the expanders is at most
    $\Of(1/\epsA^2 )$, resulting in the given bound on the number of
    edges.  There are also the additional pairs between a node in
    $\Tree$ and its parent, but since every point participates in only
    $\Of(\log{\Phi(\PS)})$ such pairs, the number of edges is
    dominated by the expanders on the \WSPD pairs.  It follows that
    the number of edges in the resulting graph is
    $\Of(\epsA^{-2} \eps^{-d}n\log{\Phi(\PS)})$. 
\end{proof}

We note that there are point sets, with unbounded spread, such that
any \WSPD on them has weight $\Omega(n^2)$. A simple example is a
sequence of points along a line with exponentially increasing
distances. Thus, requiring the bounded spread on the point set is
unavoidable to achieve optimal size by using the above construction.

\begin{definition}
    For a number $\gamma \in (0,1)$, and failed set of vertices
    $\BSet\subseteq \PS$, a node $v$ of the quadtree $\Tree$ is in the
    \emph{$\gamma$-shadow} if
    $\cardin{\BSet \cap \PcellX{v}} \geq \gamma
    \cardin{\PcellX{v}}$. Naturally, if $v$ is in the $\gamma$-shadow,
    then the points of $\PcellX{v}$ are also in the shadow. As such,
    the \emphi{$\gamma$-shadow} of $\BSet$ is the set of all the
    points in the shadow -- formally,
    \begin{math}
        \ShadowY{\gamma}{\BSet}%
        =%
        \bigcup_{v\in T \colon \cardin{\BSet\cap \PcellX{v}} \geq
           \gamma\cardin{\PcellX{v}}}\PcellX{v}.
    \end{math}
\end{definition}%

Let $\gamma = 1-\epsR/2$.  Note that
$\BSet \subseteq \ShadowY{\gamma}{\BSet}$, since every point of
$\BSet$ is stored as a singleton in a leaf of $\Tree$.

\begin{definition}
    For a node $x$ in $\Tree$, let $\sizeX{x} = \cardin{\PcellX{x}}$,
    and $\sizeBadX{x} = \cardin{\PcellX{x} \cap \BSet}$.
\end{definition}

\begin{lemma}
    \lemlab{shadow-2d}%
    Let $\gamma = 1-\epsR/2$ and $\BSet\subseteq \PS$ be fixed. Then,
    the size of the $\gamma$-shadow of $\BSet$ is at most
    $(1+\epsR)\cardin{\BSet}$.
\end{lemma}
\begin{proof}
    Let $H$ be the set of nodes of $\Tree$ that are in the
    $\gamma$-shadow of $\BSet$.  A node $ u \in H$ is \emph{maximal}
    if none of its ancestors is in $H$. Let
    $H' = \brc{u_1,\ldots, u_m}$ be the set of all maximal nodes in
    $H$, and observe that
    \begin{math}
        \bigcup_{u\in H'} \PcellX{u} = \bigcup_{v\in H} \PcellX{v} =
        \ShadowY{\gamma}{\BSet}.
    \end{math}
    For any two nodes $x,y\in H'$, we have
    $\PcellX{x} \cap \PcellX{y} = \varnothing$. Therefore, we have
    \begin{equation*}
        \cardin{\BSet}%
        =%
        \sum_{u \in H'} \sizeBadX{u}%
        \geq%
        \sum_{u \in H'} \gamma  \sizeX{u}%
        =%
        \gamma  \cardin{\ShadowY{\gamma}{\BSet}}.
    \end{equation*}
    Dividing both sides by $\gamma$ implies the claim, since
    $1/\gamma = 1 / (1-\epsR/2) \leq 1+\epsR$. 
\end{proof}

\begin{lemma}
    \lemlab{climbup}%
    Let $\gamma = 1-\epsR/2$.  Fix a node $u \in \Tree$ of the
    quadtree, the failure set $\BSet\subseteq \PS$, its shadow
    $\EBSet = \ShadowY{\gamma}{\BSet}$, and the residual graph
    $\GraphA = \GSB \setminus \BSet$. For a point
    $\pp \in \PcellX{u} \setminus \EBSet$, we define the set
    \begin{equation*}
        \rchY{u}{\pp}%
        =%
        \Set{ \pq \in \PcellX{u} \setminus \BSet}%
        { \dGZ{\GraphA}{\pp}{\pq} \leq 2\cdot\diamX{\cellX{u}}
           \bigr.}.%
    \end{equation*}
    Then, we have
    \begin{math}
        \cardin{\rchY{u}{\pp}}%
        \geq%
        3\epsA \cardin{\PcellX{u}}.
    \end{math}
\end{lemma}
\begin{proof}
    Let $u_1,u_2,\dots,u_j=u$ be the sequence of nodes in the quadtree
    from the leaf $u_1$ that contains (only) $\pp$, to the node $u$.
    A \emph{level} of a point $\pq \in \PcellX{u}$, denoted by
    $\levelX{\pq}$, is the first index $i$, such that
    $\pq \in \PcellX{u_i}$.  A \emph{skipping path} in $\GSB$, is a
    sequence of edges
    $\pp \pq_1, \pq_1 \pq_2, \ldots \pq_{m-1} \pq_{m}$, such that
    $\levelX{\pq_i} < \levelX{\pq_{i+1}}$, for all $i$.

    Let $\orchX{i}$ be the set of all points in
    $\PcellX{u_i} \setminus \BSet$ that are reachable by a skipping
    path in $\GraphA$ from $\pp$.  We claim, for $i=1,\ldots, j$, that
    \begin{equation*}
        \cardin{\orchX{i}}%
        \geq
        (1-\epsA) \sizeX{u_i} -
        \sizeBadX{u_i}
        \geq
        (1 - \epsA - \gamma) \sizeX{u_{i}}
        =%
        (\epsR/2 - \epsA ) \sizeX{u_{i}}
        =%
        3\epsA \sizeX{u_i},
    \end{equation*}
    since $\epsA = \epsR/8$ and $\pp$ is not in the $\gamma$-shadow.
    The claim clearly holds for $u_1$. So, assume inductively that the
    claim holds for $u_1, \ldots, u_{j-1}$.  Let $v_1, \ldots, v_m$ be
    the children of $u_j$ that have points stored in them (excluding
    $u_{j-1}$). There is an expander between $\PcellX{u_{j-1}}$ and
    $\PcellX{v_i}$, for all $i$, as a subgraph of $\GSB$.  It follows,
    by induction, that
    \begin{align*}
      \cardin{\orchX{j}}
      &\geq %
        (1-\epsA) \sizeX{u_{j-1}} -
        \sizeBadX{u_{j-1}}
        +
        \sum_i \pth{ (1-\epsA) \sizeX{v_i}  -
        \sizeBadX{v_{i}} \bigr.}
      \\&
      =%
      (1-\epsA) \sizeX{u_{j-1}}
      +
      \sum_i  (1-\epsA) \sizeX{v_i}
      -\Bigl( \sizeBadX{u_{j-1}} + \sum_i \sizeBadX{v_i} \Bigr)
      \\&
      =
      (1-\epsA) \sizeX{u_{j}}
      -
      \sizeBadX{u_{j}}.
    \end{align*}

    Observe that a skipping path from $\pp$ to $\pq \in \PcellX{u_j}$
    has length at most
    \begin{equation*}
        \sum_{i=1}^j \diamX{\cellX{u_i}}%
        \leq %
        \diamX{\cellX{u_j}} \sum_{i=1}^j 2^{1-j}
        \leq%
        2 \cdot \diamX{\cellX{u_j}}.
    \end{equation*}
    Thus, $ \orchX{j} \subseteq \rchY{u}{\pp}$, and the claim follows. 
\end{proof}

Now we are ready to prove that $\GSB$ is a reliable spanner.

\begin{lemma}
    \lemlab{G_2d-bdd}%
    For a set $\PS \subseteq \Re^d$ of $n$ points and parameters
    $\eps \in (0,1)$ and $\epsR \in (0,1/2)$, the graph $\GSB$ is a
    $\epsR$-reliable $(1+\eps)$-spanner with
    $\Of\bigl( \eps^{-d} \epsR^{-2} n \log \SpreadC(\PS) \bigr)$
    edges, where $ \SpreadC(\PS)$ is the spread of $\PS$.
\end{lemma}

\begin{proof}
    Let $\epsA=\epsR/8$ and $\gamma = 1 - \epsR/2$.  The bound on the
    number of edges follows by \lemref{G_2d-size}.

    Let $\BSet$ be a set of faulty vertices of $\GSB$, and let
    $\GraphA = \GSB \setminus \BSet$ be the residual graph. We define
    $\EBSet$ to contain the vertices that are in the $\gamma$-shadow
    of $\BSet$. Then, we have $\BSet \subseteq \EBSet$ and
    $\cardin{\EBSet} \leq (1+\epsR) \cardin{\BSet}$ by
    \lemref{shadow-2d}.  Finally, we need to show that there exists a
    $(1+\eps)$-path between any $\pp,\pq\in \PS \setminus \EBSet$.

    Let $\{u,v\}\in\WS$ be the pair that separates $\pp$ and $\pq$
    with $\pp \in \PcellX{u}$ and $\pq \in \PcellX{v}$, see
    \figref{1+eps-path}.  Let $\rchY{u}{\pp}$ (resp. $\rchY{v}{\pq}$)
    be the set of points in $\PcellX{u}$ (resp. $\PcellX{v}$) that are
    reachable in $\GraphA$ from $\pp$ (resp. $\pq$) with paths that
    have lengths at most $2 \cdot \diamX{\cellX{u}}$ (resp.
    $2 \cdot \diamX{\cellX{v}}$).  By \lemref{climbup},
    $\cardin{\rchY{u}{\pp}} \geq 3 \epsA \sizeX{u} \geq \epsA
    \sizeX{u}$ and $\cardin{\rchY{v}{\pq}} \geq 3 \epsA \sizeX{v}$.

    Since there is a bipartite expander between $\PcellX{u}$ and
    $\PcellX{v}$ with parameter $\epsA$, by \lemref{expander}, the
    neighborhood $Y$ of $\rchY{u}{\pp}$ in $\PcellX{v}$ has size at
    least $(1-\epsA) \sizeX{v}$. Observe that
    \begin{math}
        \cardin{Y \cap \rchY{v}{\pq}} =%
        \cardin{ \rchY{v}{\pq} \setminus \pth{\PcellX{v} \setminus Y}
        }%
        \geq%
        \cardin{ \rchY{v}{\pq}} - \cardin{ {\PcellX{v} \setminus Y} }
        \geq%
        3 \epsA \sizeX{v} -\epsA \sizeX{v} > 0.
    \end{math}
    Therefore, there is a point $\pq'\in Y \cap \rchY{v}{\pq}$, and a
    point $\pp'\in \rchY{u}{\pp}$, such that
    $\pp'\pq' \in \EdgesX{\GSB}$. We have that
    \begin{align*}
      \dGZ{\GraphA}{\pp}{\pq}
      &%
        \leq
        \dGZ{\GraphA}{\pp}{\pp'} +
        \dGZ{\GraphA}{\pp'}{\pq'} +
        \dGZ{\GraphA}{\pq'}{\pq}
      \\&%
      \leq%
      2 \cdot\diamX{\cellX{u}} + \dY{\pp'}{\pq'} + 2\cdot\diamX{\cellX{v}}
      \\&%
      \leq%
      3\cdot \diamX{\cellX{u}} + \distSetY{\cellX{u}}{\cellX{v}} +
      3\cdot\diamX{\cellX{v}}
      \\&%
      \leq%
      \pth{1 + 6 \cdot\frac{\eps}{6}}\cdot \distSetY{\cellX{u}}{\cellX{v}}
      \\&%
      \leq%
      \pth{1 + \eps}\cdot \dY{\pp}{\pq}. 
    \end{align*} 
\end{proof}

\begin{figure}[t]
    \centerline{\includegraphics[width=0.55\linewidth]%
       {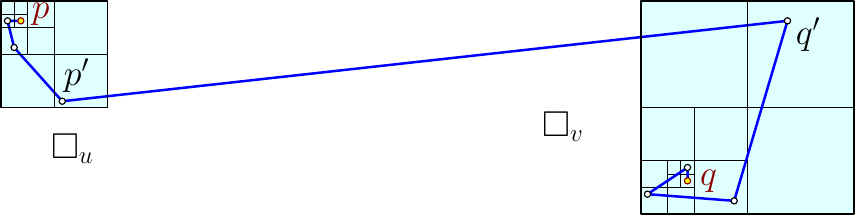}}%
    \caption{The pair $\{u,v\}\in\WS$ that separates $\pp$ and
       $\pq$. The blue path is a $(1+\eps)$-path between $\pp$ and
       $\pq$ in the graph $\GSB\setminus \BSet$.}
    \figlab{1+eps-path}
\end{figure}

\section{Conclusions}
\seclab{conclusions}

A natural open question is whether $\epsR$-reliable spanners can be
constructed with $\Of(n\log n)$ edges for general point sets. There
are different approaches that lead to near optimal bounds. While we
use the one-dimensional construction with a careful application of
locality-sensitive orderings, Bose \etal \cite{bcdm-norgm-18} uses
\WSPD, centroid decomposition and an idea of Willard
\cite{w-mdsfde-82} for order maintenance. Another natural open
question is how to construct reliable spanners that are required to be
subgraphs of a given graph.

\subsection*{Acknowledgements.}
The authors thank Pat Morin for his useful comments on this
manuscript, and keeping us updated about his own work
\cite{bcdm-norgm-18}. The authors also thank the anonymous referees
for the detailed and useful comments.

%*flatex input: [./shadow.bbl]
\newcommand{\etalchar}[1]{$^{#1}$}
 \providecommand{\CNFWADS}{\CNFX{WADS}}  \providecommand{\CNFX}[1]{
  {\em{\textrm{(#1)}}}}  \providecommand{\CNFISAAC}{\CNFX{ISAAC}}
  \providecommand{\CNFCCCG}{\CNFX{CCCG}}

% flatex input end: [./shadow.bbl]

\end{document}